\documentclass[letterpaper, 12pt]{article}[2000/05/19]
\usepackage[english]{babel}
\usepackage{amsfonts,amsmath,amssymb,amsthm,latexsym,amscd,mathrsfs}
\usepackage{ifthen,cite}
\usepackage[bookmarksnumbered=true]{hyperref}

\hypersetup{pdfpagetransition={Split}}

\newcommand{\evenhead}{Author \ name}
\newcommand{\oddhead}{Article \ name}
\newcommand{\theArticleName}{Article name}

\newcommand{\FirstPageHeading}[1]{\thispagestyle{empty}%
\noindent\raisebox{0pt}[0pt][0pt]{\makebox[\textwidth]{\protect\footnotesize \sf }}\par}

\newcommand{\ArticleName}[1]{\renewcommand{\theArticleName}{#1}\vspace{-2mm}\par\noindent {\LARGE\bf  #1\par}}
\newcommand{\Author}[1]{\vspace{5mm}\par\noindent {\Large  #1\par} \par\vspace{2mm}\par}
\newcommand{\Address}[1]{\vspace{2mm}\par\noindent {\it #1} \par}
\newcommand{\Email}[1]{\ifthenelse{\equal{#1}{}}{}{\par\noindent {\rm E-mail: }{\it  #1} \par}}
\newcommand{\URLaddress}[1]{\ifthenelse{\equal{#1}{}}{}{\par\noindent {\rm URL: }{\tt  #1} \par}}
\newcommand{\EmailD}[1]{\ifthenelse{\equal{#1}{}}{}{\par\noindent {$\phantom{\dag}$~\rm E-mail: }{\it  #1} \par}}
\newcommand{\URLaddressD}[1]{\ifthenelse{\equal{#1}{}}{}{\par\noindent {$\phantom{\dag}$~\rm URL: }{\tt  #1} \par}}

\newcommand{\Abstract}[1]{\vspace{6mm}\par\noindent\hspace*{10mm}
\parbox{140mm}{\small {\bf Abstract.} #1}\par}
\newcommand{\Keywords}[1]{\vspace{3mm}\par\noindent\hspace*{10mm}
\parbox{140mm}{\small {\bf Key words:} \rm #1}\par}
\newcommand{\Classification}[1]{\vspace{3mm}\par\noindent\hspace*{10mm}
\parbox{140mm}{\small {\it 2000 Mathematics Subject Classification:} \rm #1}\vspace{3mm}\par}
\newcommand{\ShortArticleName}[1]{\renewcommand{\oddhead}{#1}}
\newcommand{\AuthorNameForHeading}[1]{\renewcommand{\evenhead}{#1}}

\long\def\@makecaption#1#2{
  \sbox\@tempboxa{\small \textbf{#1.}\ \ #2}%
  \ifdim \wd\@tempboxa >\hsize
    {\small \textbf{#1.}\ \ #2}\par \else
    \global \@minipagefalse
    \hb@xt@\hsize{\hfil\box\@tempboxa\hfil}%
  \fi \vskip\belowcaptionskip}

\def\numberwithin#1#2{\@ifundefined{c@#1}{\@nocounterr{#1}}{%
  \@ifundefined{c@#2}{\@nocnterr{#2}}{%
  \@addtoreset{#1}{#2}%
  \toks@\@xp\@xp\@xp{\csname the#1\endcsname}%
  \@xp\xdef\csname the#1\endcsname
    {\@xp\@nx\csname the#2\endcsname.\the\toks@}}}}
\def\E^#1{{\buildrel #1 \over\vee}}
\newtheorem{theorem}{Theorem}

{\theoremstyle{definition}

}

\begin{document}

\FirstPageHeading{V.I. Gerasimenko}

\ShortArticleName{Enskog Equation}

\AuthorNameForHeading{I.V. Gapyak, V.I. Gerasimenko}

\ArticleName{On Rigorous Derivation of the Enskog Kinetic Equation}

\Author{I.V. Gapyak$^\ast$\footnote{E-mail: \emph{gapjak@ukr.net}}
        and V.I. Gerasimenko$^\ast$$^\ast$\footnote{E-mail: \emph{gerasym@imath.kiev.ua}}}

\Address{$^\ast$\hspace*{1mm}Taras Shevchenko National University of Kyiv,\\
    \hspace*{3mm}Department of Mechanics and Mathematics,\\
    \hspace*{3mm}2, Academician Glushkov Av.,\\
    \hspace*{3mm}03187, Kyiv, Ukraine}

\Address{$^\ast$$^\ast$Institute of Mathematics of NAS of Ukraine,\\
    \hspace*{4mm}3, Tereshchenkivs'ka Str.,\\
    \hspace*{4mm}01601, Kyiv-4, Ukraine}

\bigskip

\Abstract{We develop a rigorous formalism for the description of the kinetic evolution of infinitely many
hard spheres. On the basis of the kinetic cluster expansions of cumulants of groups of operators of finitely
many hard spheres the nonlinear kinetic Enskog equation and its generalizations are justified. It is
established that for initial states which are specified in terms of one-particle distribution functions
the description of the evolution by the Cauchy problem of the BBGKY hierarchy and by the Cauchy problem of
the generalized Enskog kinetic equation together with a sequence of explicitly defined functionals of a
solution of stated kinetic equation are an equivalent. For the initial-value problem of the generalized Enskog
equation the existence theorem is proved in the space of integrable functions.}

\bigskip

\vspace{2pc}
\Classification{35Q20; 35Q82; 47J35; 82C05; 82C40.}

\makeatletter
\renewcommand{\@evenhead}{
\hspace*{-3pt}\raisebox{-15pt}[\headheight][0pt]{\vbox{\hbox to \textwidth {\thepage \hfil \evenhead}\vskip4pt \hrule}}}
\renewcommand{\@oddhead}{
\hspace*{-3pt}\raisebox{-15pt}[\headheight][0pt]{\vbox{\hbox to \textwidth {\oddhead \hfil \thepage}\vskip4pt\hrule}}}
\renewcommand{\@evenfoot}{}
\renewcommand{\@oddfoot}{}
\makeatother

\newpage
\vphantom{math}
\protect\tableofcontents

\vspace{0.5cm}

\section{Introduction}

Nowadays the considerable advance \cite{CGP97,CIP,Sp91,Pe08,L} in the rigorous derivation of the
Boltzmann kinetic equation in the Boltzmann-Grad scaling limit \cite{GH,Sp80} is well known. The
lack of similar progress for the Enskog kinetic equation \cite{BLPT,Pol,Pol1,MP,BL} suggested by
D. Enskog \cite{E} as a genelalization of the Boltzmann equation on dense gases, is conditioned
by a priori stated collision integral of this kinetic equation. Thus, the rigorous derivation of
the Enskog-type equations from dynamics of hard spheres \cite{CGP97},\cite{GP85} remains an open
problem so far.

In this paper we consider the problem of potentialities inherent in the description of the evolution
of states of a hard sphere system in terms of a one-particle distribution function. In view of
the generic initial data of hard spheres are determined by initial correlations connected with
the forbidden configurations of particles, one of the problems consists in the derivation of the
kinetic equations from underlaying dynamics in case of presence of correlations at initial time.
We establish that in fact, if initial data is completely defined by a one-particle distribution
function, then all possible states of infinitely many hard spheres at arbitrary moment of time
can be described within the framework of a one-particle distribution function by the kinetic
equation without any approximations.

We briefly outline the main results and structure of the paper.
In section 2 we formulate some preliminary facts about dynamics of finitely many hard spheres
and state the main result.
In sections 3 and 4 the main results related to the origin of the Enskog kinetic evolution are
proved. On the basis of the kinetic cluster expansions of cumulants of groups of operators of
finitely many hard spheres which give the possibility to construct the kinetic equations in case
of presence of correlations at initial time, in section 3 we derive the generalized Enskog equation.
Then in section 4 we establish that all possible states are described by a sequence of explicitly
defined functionals of a solution of constructed kinetic equation. Moreover, in section 5 we state
the Markovian generalized Enskog equation and establish its links with the Enskog-type kinetic equations
\cite{BLPT},\cite{BE,RL,B46}. In section 6 the existence theorem is proved for the abstract initial-value
problem of the generalized Enskog equation in the space of integrable functions.
Finally, in section 7 we conclude with some observations and perspectives for future research.


\section{On the origin of the Enskog kinetic evolution}

\subsection{Preliminaries: dynamics of finitely many hard spheres}
We consider a system of a non-fixed, i.e. arbitrary but finite, number of particles of a unit mass
(the formalism of nonequilibrium grand canonical ensemble \cite{CGP97}), interacting as hard spheres
with a diameter $\sigma>0$. Every particle is characterized by the phase coordinates
$(q_{i},p_{i})\equiv x_{i}\in\mathbb{R}^{3}\times\mathbb{R}^{3},\, i\geq1.$ For configurations
of such a system the following inequalities hold: $|q_i-q_j|\geq\sigma,\,i\neq j\geq1$, i.e.
the set $\mathbb{W}_n\equiv\big\{(q_1,\ldots,q_n)\in\mathbb{R}^{3n}\big||q_i-q_j|<\sigma$ for
at least one pair $(i,j):\,i\neq j\in(1,\ldots,n)\big\}$ is the set of forbidden configurations
of $n$ hard spheres.

The evolution of all possible states is described by the sequence of marginal ($s$-particle)
distribution functions $F_s(t,x_1,\ldots,x_s),\,s\geq1,$ that satisfy the initial-value problem
of the BBGKY hierarchy \cite{CGP97}
\begin{eqnarray}
  \label{NelBog1}
    &&\frac{\partial}{\partial t}F_s(t)=\mathcal{L}^*_{s}F_{s}(t)+
       \sum_{i=1}^{s}\int_{\mathbb{R}^{3}\times\mathbb{R}^{3}}dx_{s+1}
       \mathcal{L}^*_{\mathrm{int}}(i,s+1)F_{s+1}(t),\\ \nonumber\\
  \label{eq:NelBog2}
    &&F_s(t)_{\mid t=0}=F_s^0, \quad s\geq1.
\end{eqnarray}
If $t>0$, in hierarchy of evolution equations (\ref{NelBog1}) the operator $\mathcal{L}^{*}_{s}$
is defined by the Poisson bracket of noninteracting particles with the corresponding boundary
conditions on $\partial\mathbb{W}_{s}$ \cite{CGP97}:
\begin{eqnarray}\label{OperL}
   &&\mathcal{L}^{*}_{s}F_{s}(t)\doteq-\sum\limits_{i=1}^{s}
     \langle p_{i},\frac{\partial}{\partial q_{i}}\rangle_{\mid_{\partial \mathbb{W}_{s}}}F_{s}(t,x_1,\dots,x_s),
\end{eqnarray}
where $\langle\eta,(p_i-p_{s+1})\rangle\doteq{\sum}_{\alpha=1}^3\eta^{\alpha}(p_i^\alpha-p_{s+1}^\alpha)$
is a scalar product. Operator (\ref{OperL}) is the generator of the Liouville equation for states and it
is an adjoint operator to the generator $\mathcal{L}_{s}$ of the Liouville equation for observables  \cite{CGP97}.
The operator $\mathcal{L}^*_{\mathrm{int}}(i,s+1)$ is defined by the expression
\begin{eqnarray}\label{aLint}
   &&\sum_{i=1}^{s}\int_{\mathbb{R}^{3}\times\mathbb{R}^{3}}dx_{s+1}
     \mathcal{L}^{*}_{\mathrm{int}}(i,s+1)F_{s+1}(t)=
     \sigma^{2}\sum\limits_{i=1}^s\int_{\mathbb{R}^3\times\mathbb{S}_{+}^{2}}
     d p_{s+1}d\eta\,\langle\eta,(p_i-p_{s+1})\rangle\times\\
   &&\times\big(F_{s+1}(t,x_1,\ldots,q_i,p_i^{*},\ldots,x_s,q_i-\sigma\eta,p_{s+1}^{*})
     -F_{s+1}(t,x_1,\ldots,x_s,q_i+\sigma\eta,p_{s+1})\big),\nonumber
\end{eqnarray}
where ${\Bbb S}_{+}^{2}\doteq\{\eta\in\mathbb{R}^{3}\big|\,|\eta|=1,\langle\eta,(p_i-p_{s+1})\rangle>0\}$
and the momenta $p_{i}^{*},$ $p_{s+1}^{*}$ are defined by the following equalities
\begin{eqnarray} \label{eq:momenta}
   &&p_i^{*}\doteq p_i-\eta\,\left\langle\eta,\left(p_i-p_{s+1}\right)\right\rangle, \\
   &&p_{s+1}^{*}\doteq p_{s+1}+\eta\,\left\langle\eta,\left(p_i-p_{s+1}\right)\right\rangle \nonumber.
\end{eqnarray}
If $t<0$, the generator of the BBGKY hierarchy is defined by the corresponding expression \cite{CGP97}.

We will consider initial data (\ref{eq:NelBog2}) satisfying the chaos condition \cite{CGP97}, i.e.
statistically independent hard spheres
\begin{eqnarray}\label{eq:Bog2_haos}
    &&F_s(t)_{\mid t=0}=\prod_{i=1}^{s}F_{1}^0(x_i)\mathcal{X}_{\mathbb{R}^{3s}\setminus \mathbb{W}_s},
\end{eqnarray}
where $\mathcal{X}_{\mathbb{R}^{3s}\setminus \mathbb{W}_s}\equiv\mathcal{X}_s(q_1,\ldots,q_s)$ is a
characteristic function of allowed configurations $\mathbb{R}^{3s}\setminus \mathbb{W}_s$ of $s$ hard
spheres. Initial data (\ref{eq:Bog2_haos}) is intrinsic for the kinetic description of many-particle
systems because in this case all possible states are characterized by means of a one-particle marginal
distribution function.

To construct a solution of initial-value problem (\ref{NelBog1})-(\ref{eq:NelBog2}) we adduce some
preliminaries about hard sphere dynamics.
Let $L^{1}_{n}\equiv L^{1}(\mathbb{R}^{3n}\times(\mathbb{R}^{3n}\setminus \mathbb{W}_n))$ be the
space of integrable functions that are symmetric with respect to the permutations of the arguments
$x_1,\ldots,x_n$, equal to zero on the set of forbidden configurations $\mathbb{W}_n$ and equipped
with the norm: $\|f_n\|_{L^{1}(\mathbb{R}^{3n}\times\mathbb{R}^{3n})}=\int dx_1\ldots dx_n|f_n(x_1,\ldots,x_n)|$.
We denote by $L_{n,0}^1\subset L^1_n$ the subspace of continuously differentiable functions
with compact supports.

On a set of measurable functions the following one-parameter mapping: $\mathbb{R}\ni t\mapsto S_{n}(-t)f_{n}$
is defined by
\begin{eqnarray} \label{Sspher}
  &&\hskip-5mm(S_{n}(-t)f_{n})(x_{1},\ldots,x_{n})\equiv S_{n}(-t,1,\ldots,n)f_{n}(x_{1},\ldots,x_{n})\doteq\\
  &&\hskip-5mm\doteq\begin{cases}
         f_{n}(X_{1}(-t,x_{1},\ldots,x_{n}),\ldots,X_{n}(-t,x_{1},\ldots,x_{n})),
         \hskip+4mm\mathrm{if}\,(x_{1},\ldots,x_{n})\in(\mathbb{R}^{3n}\times(\mathbb{R}^{3n}\setminus\mathbb{W}_{n})),\\
         0, \hskip+84mm \mathrm{if}\,(q_{1},\ldots,q_{n})\in\mathbb{W}_{n},
                    \end{cases}\nonumber
\end{eqnarray}
where $X_{i}(-t)$ is a phase trajectory of $ith$ particle constructed in \cite{CGP97}. We note that
the phase trajectories of a hard sphere system are determined almost everywhere on the phase space
$\mathbb{R}^{3n}\times(\mathbb{R}^{3n}\setminus \mathbb{W}_n)$ beyond the set $\mathbb{M}_{n}^0$ of
the zero Lebesgue measure \cite{CGP97,GP85}. The set $\mathbb{M}_{n}^0$ consists of the phase space
points which are specified such initial data that during the evolution generate multiple collisions,
i.e. collisions of more than two particles, more than one two-particle collision at the same instant
and infinite number of collisions within a finite time interval.

On the space $L^{1}_{n}$ mapping (\ref{Sspher}) defines the isometric strong continuous group of
operators \cite{CGP97}, and on $L_{n,0}^1\subset L^1_n$ the infinitesimal generator of group (\ref{Sspher})
is operator (\ref{OperL}) which is the generator of the Liouville equation for hard spheres \cite{CGP97}.

If $F_{1}^0\in L^{1}_{n}$, then a nonperturbative solution of initial-value problem (\ref{NelBog1}),
(\ref{eq:Bog2_haos}) is a sequence of distribution functions $F_s(t,x_1,\ldots,x_s),\,s\geq1$,
represented by the following series \cite{CGP97},\cite{GRS04}
\begin{eqnarray}\label{F(t)}
   &&F_{s}(t,x_{1},\ldots,x_{s})=\\
   &&=\sum\limits_{n=0}^{\infty}\frac{1}{n!}\int_{(\mathbb{R}^{3}\times\mathbb{R}^{3})^{n}}
      dx_{s+1}\ldots dx_{s+n}\mathfrak{A}_{1+n}(-t,\{Y\},X\setminus Y)
      \prod_{i=1}^{s+n}F_{1}^0(x_i)\mathcal{X}_{\mathbb{R}^{3(s+n)}\setminus\mathbb{W}_{s+n}},\nonumber
\end{eqnarray}
where the evolution operator $\mathfrak{A}_{1+n}(-t)$ is the $(n+1)th$-order cumulant of groups
of operators (\ref{Sspher}) defined by the expansion
\begin{eqnarray}\label{nLkymyl}
   &&\hskip-7mm\mathfrak{A}_{1+n}(-t,\{Y\},X\setminus Y)=\sum\limits_{\texttt{P}:\,(\{Y\},\,X\backslash Y)={\bigcup\limits}_i X_i}
      (-1)^{|\texttt{P}|-1}(|\texttt{P}|-1)!\prod_{X_i\subset \texttt{P}}S_{|\theta(X_i)|}(-t,\theta(X_i)),
\end{eqnarray}
and the following notations are used: $\{Y\}$ is a set consisting of one element
$Y\equiv(1,\ldots,s)$, i.e. $|\{Y\}|=1$, $\sum_\texttt{P}$ is a sum over all possible
partitions $\texttt{P}$ of the set $(\{Y\},X\setminus Y)\equiv(\{Y\},s+1,\ldots,s+n)$
into $|\texttt{P}|$ nonempty mutually disjoint subsets $X_i\in(\{Y\},X\setminus Y)$,
the mapping $\theta$ is the declasterization mapping defined by the formula:
$\theta(\{Y\},X\setminus Y)=X$. The simplest cumulants (\ref{nLkymyl}) have the form
\begin{eqnarray*}
    &&\mathfrak{A}_{1}(-t,\{Y\})=S_{s}(-t,Y),\\
    &&\mathfrak{A}_{2}(-t,\{Y\},s+1)=S_{s+1}(-t,Y,s+1)-S_{s}(-t,Y)S_{1}(-t,s+1).
\end{eqnarray*}

If $\|F_{1}^0\|_{L^{1}(\mathbb{R}^{3}\times\mathbb{R}^{3})}<e^{-1}$, series (\ref{F(t)})
converges in the norm of the space $L^{1}_{n}$ for arbitrary $t\in\mathbb{R}$. We remark
that this condition means the condition on average number of particles. If we renormalize
the distribution functions $F_s(t)=\frac{1}{v^s}\widetilde{F}_s(t)$, where $\frac{1}{v}$
is the density of particles (the average number of particles in a unit volume), we obtain
solution expansion (\ref{F(t)}) over powers of the density $\frac{1}{v}$ which converges provided
that: $\frac{1}{v}<e^{-1}\|\widetilde{F}_1^0\|_{L^{1}(\mathbb{R}^{3}\times\mathbb{R}^{3})}^{-1}$.

We note that nonperturbative solution (\ref{F(t)}) of the BBGKY hierarchy (\ref{NelBog1})
is transformed to the form of the perturbation (iteration) series as a result of applying
of analogs of the Duhamel equation \cite{BA} to cumulants of groups of operators (\ref{nLkymyl}).

In case of initial data $F_{1}^0\in L^{1}_{n,0}$ sequence of functions (\ref{F(t)}) is
a strong solution of initial-value problem (\ref{NelBog1}),(\ref{eq:Bog2_haos}) and for
arbitrary initial data $F_{1}^0\in L^{1}_{n}$ it is a weak solution \cite{GRS04}.

In the next section we describe the evolution of states within the framework of the
kinetic theory, i.e. in terms of a one-particle marginal distribution function.

\subsection{The kinetic evolution of hard spheres: main results}
In view of the fact that initial data (\ref{eq:Bog2_haos}) is completely determined by a
one-particle marginal distribution function on allowed configurations, the states given
in terms of the sequence $F(t)=(F_1(t,x_1),\ldots,F_s(t,x_1,\ldots,x_s),\ldots)$ of marginal
distribution functions (\ref{F(t)}) can be described within the framework of the sequence
$F(t\mid F_{1}(t))=(F_1(t,x_1),F_2(t,x_1,x_2\mid F_{1}(t)),\ldots,$ $F_s(t,x_1,\ldots,x_s\mid F_{1}(t)),\ldots)$
of the marginal functionals of the state $F_s(t,x_1,\ldots,x_s\mid F_{1}(t)),\,s\geq2$ which
are explicitly defined with respect to the solution $F_1(t,x_1)$ of the kinetic equation. We
refer to such kinetic equation of a hard sphere system as the generalized Enskog kinetic equation.

Indeed for $t\geq 0$ the one-particle distribution function (\ref{F(t)}) is a solution
of the following initial-value problem
\begin{eqnarray}
  \label{gke1}
   &&\frac{\partial}{\partial t}F_{1}(t,x_1)=-\langle p_1,\frac{\partial}{\partial q_1}\rangle F_{1}(t,x_1)+\\
   &&+\sigma^2\sum_{n=0}^{\infty}\frac{1}{n!}\int_{\mathbb{R}^3\times\mathbb{S}^2_+}d p_2 d\eta
      \int_{(\mathbb{R}^{3}\times\mathbb{R}^{3})^{n}}dx_3\ldots dx_{n+2}\langle\eta,(p_1-p_2)\rangle\times\nonumber \\
   &&\times\big(\mathfrak{V}_{1+n}(t,\{1^{*},2^{*}_{-}\},3,\ldots,n+2)
      F_1(t,q_1,p_1^{*})F_1(t,q_1-\sigma\eta,p_2^{*})\prod_{i=3}^{n+2}F_{1}(t,x_i)-\nonumber \\
   &&-\mathfrak{V}_{1+n}(t,\{1,2_{+}\},3,\ldots,n+2)F_1(t,x_1)
      F_1(t,q_1+\sigma\eta,p_2)\prod_{i=3}^{n+2}F_{1}(t,x_i)\big),\nonumber \\\nonumber \\
  \label{2}
   &&F_1(t,x_1)|_{t=0}= F_1^0(x_1).
\end{eqnarray}
In this evolution equation we use notations from definition (\ref{Sspher}) adopted to the conventional
notation of the Enskog collision integral: indices $(1^{\sharp},2^{\sharp}_{\pm})$ denote that the
evolution operator $\mathfrak{V}_{1+n}(t)$ acts on the corresponding phase points $(q_1,p_1^{\sharp})$
and $(q_1\pm\sigma\eta,p_2^{\sharp})$, and in general case the $(n+1)th$-order evolution operator
$\mathfrak{V}_{1+n}(t),\,n\geq0$, is defined by the expansion ($|Y|=2$ in case of equation (\ref{gke1}))
\begin{eqnarray}\label{skrr}
   &&\hskip-8mm\mathfrak{V}_{1+n}(t,\{Y\},X\setminus Y)\doteq n!\,\sum_{k=0}^{n}(-1)^k\,\sum_{m_1=1}^{n}\ldots
       \sum_{m_k=1}^{n-m_1-\ldots-m_{k-1}}\frac{1}{(n-m_1-\ldots-m_k)!}\times\\
   &&\hskip-8mm\times\widehat{\mathfrak{A}}_{1+n-m_1-\ldots-m_k}(t,\{Y\},s+1,\ldots,s+n-m_1-\ldots-m_k)\times\nonumber\\
   &&\hskip-8mm\times\prod_{j=1}^k\,\ \sum_{k_2^j=0}^{m_j}\ldots \sum_{k^j_{n-m_1-\ldots
       -m_j+s}=0}^{k^j_{n-m_1-\ldots-m_j+s-1}}\prod_{i_j=1}^{s+n-m_1-\ldots-m_j}
       \frac{1}{(k^j_{n-m_1-\ldots-m_j+s+1-i_j}-k^j_{n-m_1-\ldots-m_j+s+2-i_j})!}\times\nonumber\\
   &&\hskip-8mm\times\widehat{\mathfrak{A}}_{1+k^j_{n-m_1-\ldots-m_j+s+1-i_j}-k^j_{n-m_1-\ldots-m_j+s+2-i_j}}
       (t,i_{j},s+n-m_1-\ldots-m_j+1+\nonumber \\
   &&\hskip-8mm+k^j_{s+n-m_1-\ldots-m_j+2-i_j},\ldots,s+n-m_1-\ldots-m_j+k^j_{s+n-m_1-\ldots-m_j+1-i_j}),\nonumber
\end{eqnarray}
where we mean $k^j_1\equiv m_j,\,k^j_{n-m_1-\ldots-m_j+s+1}\equiv 0$, and by the operator
$\widehat{\mathfrak{A}}_{1+n}(t)$ we denote the $(1+n)th$-order scattering cumulant
\begin{eqnarray}\label{scacu}
   &&\widehat{\mathfrak{A}}_{1+n}(t,\{Y\},X\setminus Y)
      \doteq\mathfrak{A}_{1+n}(-t,\{Y\},X\setminus Y)\mathfrak{I}_{s+n}(X)
      \prod_{i=1}^{s+n}\mathfrak{A}_{1}(t,i).
\end{eqnarray}
In definition (\ref{scacu}) the operator $\mathfrak{A}_{1+n}(-t)$ is the $(1+n)th$-order cumulant
(\ref{nLkymyl}) of groups (\ref{Sspher}) and the operator $\mathfrak{I}_{s+n}$ is defined by the formula
\begin{eqnarray}\label{ind}
   &&\mathfrak{I}_{s+n}(X)f_{s+n}\doteq \mathcal{X}_{\mathbb{R}^{3(s+n)}\setminus \mathbb{W}_{s+n}}f_{s+n},
\end{eqnarray}
where $\mathcal{X}_{\mathbb{R}^{3(s+n)}\setminus \mathbb{W}_{s+n}}$ is a characteristic function of
allowed configurations of $s+n$ hard spheres.
In section 4 we will give an illustration of evolution operators (\ref{skrr}).

If $t<0$, the generalized Enskog equation has the corresponding form
\begin{eqnarray}\label{gke1n}
    &&\frac{\partial}{\partial t}F_{1}(t,x_1)=-\langle p_1,\frac{\partial}{\partial q_1}\rangle F_{1}(t,x_1)+\\
    &&+\sigma^2\sum_{n=0}^{\infty}\frac{1}{n!}\int_{\mathbb{R}^3\times\mathbb{S}^2_+}d p_2 d\eta
       \int_{(\mathbb{R}^{3}\times\mathbb{R}^{3})^{n}}dx_3\ldots dx_{n+2}\langle\eta,(p_1-p_2)\rangle\times \nonumber \\
    &&\times\big(\mathfrak{V}_{1+n}(t,\{1,2_{-}\},3,\ldots,n+2)
       F_1(t,q_1,p_1)F_1(t,q_1-\sigma\eta,p_2)\prod_{i=3}^{n+2}F_{1}(t,x_i)-\nonumber \\
    &&-\mathfrak{V}_{1+n}(t,\{1^{*},2^{*}_{+}\},3,\ldots,n+2)F_1(t,q_1,p_1^{*})
       F_1(t,q_1+\sigma\eta,p_2^{*})\prod_{i=3}^{n+2}F_{1}(t,x_i)\big).\nonumber
\end{eqnarray}

The marginal functionals of the state $F_{s}\big(t,x_{1},\ldots,x_{s} \mid F_{1}(t)\big),\,s\geq2$,
are represented by the following expansions over products of the one-particle distribution function $F_{1}(t)$
\begin{eqnarray}\label{f}
    &&F_{s}\big(t,x_{1},\ldots,x_{s}\mid F_{1}(t)\big)\doteq\\
    &&\doteq\sum _{n=0}^{\infty }\frac{1}{n!}
       \int_{(\mathbb{R}^{3}\times\mathbb{R}^{3})^{n}}dx_{s+1}\ldots dx_{s+n}\,
       \mathfrak{V}_{1+n}(t,\{Y\},X\setminus Y)\prod_{i=1}^{s+n}F_{1}(t,x_i),\nonumber
\end{eqnarray}
where the $(n+1)th$-order evolution operator $\mathfrak{V}_{1+n}(t),\,n\geq0$, is defined by expansion
(\ref{skrr}). If $\|F_{1}(t)\|_{L^{1}(\mathbb{R}^{3}\times\mathbb{R}^{3})}<e^{-(3s+2)}$, series (\ref{f})
converges in the norm of the space $L^{1}_{s}$ for arbitrary $t\in\mathbb{R}$.

We remark that in terms of marginal functional of the state (\ref{f}) the collision integral
of the generalized Enskog equation is represented by the expansion
\begin{eqnarray*}
    &&\mathcal{I}_{GEE}=\sigma^{2}\int_{\mathbb {R}^3\!\times\mathbb{S}_{+}^{2}}
       d p_{2}d\eta\,\langle\eta,(p_1-p_{2})\rangle\times\\
    &&\times\big(F_{2}(t,q_1,p_1^*,q_1-\sigma\eta,p_{2}^*\mid F_{1}(t))
        -F_{2}(t,x_1,q_1+\sigma\eta,p_{2}\mid F_{1}(t))\big).
\end{eqnarray*}

In section 5 we establish the relationship of the generalized Enskog equation (\ref{gke1}),(\ref{gke1n})
with the Markovian generalized Enskog equation and the revised Enskog equation \cite{BE,RL}. Here we
only point out that first term of the collision integral expansion in kinetic equation (\ref{gke1}) is
exactly the collision integral of the Boltzmann-Enskog equation.

Thus, under the condition: $\|F_1^0\|_{L^{1}(\mathbb{R}^{3}\times\mathbb{R}^{3})}<e^{-(3s+4)}(1+e^{-(3s+3)})^{-1}$,
initial-value problem of the BBGKY hierarchy (\ref{NelBog1}),(\ref{eq:Bog2_haos}) is equivalent to
initial-value problem of the generalized Enskog equation (\ref{gke1})-(\ref{2}) and the sequence of
marginal functionals of the state $F_{s}\big(t,x_{1},\ldots,x_{s}\mid F_{1}(t)\big),\,s\geq2$, defined
by expansions (\ref{f}).
The proof of this statement is the subject of the next two sections. The existence theorem of the Cauchy
problem (\ref{gke1})-(\ref{2}) in the space of integrable functions is stated in section 6.

Finally, it should be noted that the possibility to describe the evolution of all possible states of
hard spheres in an equivalent way by the restated Cauchy problem for the generalized Enskog equation
and by a sequence of the marginal functionals of the state as compared with the corresponding Cauchy
problem of the BBGKY hierarchy is an inherent property of the description of many-particle systems
within the framework of the formalism of nonequilibrium grand canonical ensemble which adopted to the
description of infinite-particle systems in suitable functional spaces \cite{CGP97}.

\subsection{A one-dimensional hard sphere system}
We consider a one-dimensional system of hard spheres (hard rods of the length $\sigma>0$). In view
of that in this case for $t>0$ operator (\ref{aLint}) is defined by the expression \cite{Ger}
\begin{eqnarray*}\label{aLint1}
    &&\hskip-7mm\sum_{i=1}^{s}\int_{\mathbb{R}^{3}\times\mathbb{R}^{3}}dx_{s+1}
       \mathcal{L}^*_{\mathrm{int}}(i,s+1)F_{s+1}(t)=\\
    &&\hskip-7mm\sum\limits_{i=1}^s\int_{0}^{\infty}d P P\big(F_{s+1}(t,x_1,\ldots,q_i,p_i-P,\ldots,x_s,q_i-\sigma,p_i)-
       F_{s+1}(t,x_1,\ldots,x_s,q_i-\sigma ,p_i+P)+\nonumber \\
    &&\hskip-5mm+F_{s+1}(t,x_1,\ldots,q_i,p_i+P,\ldots,x_s,q_i+\sigma,p_i)-
       F_{s+1}(t,x_1,\ldots,x_s,q_i+\sigma,p_i-P)\big),\nonumber
\end{eqnarray*}
then for $t>0$ the generalized Enskog kinetic equation for a one-dimensional hard sphere system has
the form
\begin{eqnarray*}\label{OneEnsEq1}
  &&\hskip-5mm\frac{\partial}{\partial t}F_{1}(t,x_1)=-p_1\frac{\partial}{\partial q_1}F_{1}(t,x_1)+\\
  &&\hskip-5mm+\sum\limits_{n=0}^{\infty}\frac{1}{n!}\int_{0}^{\infty}d P P
     \big(F_{2}(t,q_1,p_1-P,q_1-\sigma,p_1\mid F_{1}(t))-F_{2}(t,q_1,p_1,q_1-\sigma,p_1+P\mid F_{1}(t))+ \\
  &&\hskip-5mm+F_{2}(t,q_1,p_1+P,q_1+\sigma,p_1\mid F_{1}(t))-
     F_{2}(t,q_1,p_1,q_1+\sigma,p_1-P\mid F_{1}(t))\big),\nonumber
\end{eqnarray*}
where the marginal functional of the state $F_{2}\big(t,x_1,x_2\mid F_{1}(t)\big)$ is defined by formula
(\ref{f}). The series in right-hand side of this equation converges if
$\|F_1(t)\|_{L^1(\mathbb{R}\times\mathbb{R})}<e^{-8}$ (see section 4). For $t<0$, the structure of the
collision integral of the generalized Enskog equation induces by the structure of the corresponding term
of the BBGKY hierarchy for a one-dimensional hard sphere system \cite{Ger}.


\section{The generalized Enskog kinetic equation}

\subsection{Auxiliary facts: kinetic cluster expansions}
To reformulate the Cauchy problem (\ref{NelBog1}),(\ref{eq:Bog2_haos}) as the new Cauchy problem
for a one-particle distribution function governed by the kinetic equation and the sequence of
explicitly defined marginal functionals of the state which determined by a solution of such
Cauchy problem, we introduce the following cluster expansions of cumulants (\ref{nLkymyl}) of
groups of operators (\ref{Sspher})
\begin{eqnarray}\label{rrrl2}
   &&\hskip-5mm\mathfrak{A}_{1+n}(-t,\{Y\},X\setminus Y)\mathfrak{I}_{s+n}(X)=\\ \nonumber
   &&\hskip-5mm=\sum_{k_1=0}^{n}\frac{n!}{(n-k_1)!k_1!}\,
     \mathfrak{V}_{1+n-k_1}(t,\{Y\},s+1,\ldots,s+n-k_1)\times \nonumber
\end{eqnarray}
\begin{eqnarray*}
   &&\hskip-5mm\times\sum_{k_2=0}^{k_1}\frac{k_1!}{k_2!(k_1-k_2)!}\ldots\sum_{k_{n-k_1+s}=0}^{k_{n-k_1+s-1}}
     \frac{k_{n-k_1+s-1}!}{k_{n-k_1+s}!(k_{n-k_1+s-1}-k_{n-k_1+s})!}\times \\ \nonumber
   &&\hskip-5mm\times\prod_{i=1}^{s+n-k_1}{\mathfrak{A}}_{1+k_{n-k_1+s+1-i}-k_{n-k_1+s+2-i}}
     (-t,i,s+n-k_1+1+k_{s+n-k_1+2-i},\ldots\\ \nonumber
   &&\hskip-5mm\ldots,s+n-k_1+k_{s+n-k_1+1-i})
     \mathfrak{I}_{1+k_{n-k_1+s+1-i}-k_{n-k_1+s+2-i}}(i,s+n-k_1+1+\nonumber\\
   &&\hskip-5mm+k_{s+n-k_1+2-i},\ldots,s+n-k_1+k_{s+n-k_1+1-i}), \qquad n\geq0,\nonumber
\end{eqnarray*}
where the operator $\mathfrak{I}_{s+n}$ is defined by formula (\ref{ind}) and the following convention
is assumed: $k_{s+1}\equiv 0$. In case of initial states (\ref{eq:Bog2_haos}) that involve correlations
such type (\ref{rrrl2}) cluster expansions (the kinetic cluster expansions) permits to take into consideration
initial correlations in the kinetic equation.

We give a few examples of recurrence relations (\ref{rrrl2}) in terms of scattering cumulants
(\ref{scacu}). Acting on both sides of equality (\ref{rrrl2}) by the evolution operators
$\prod_{i=1}^{s+n}\mathfrak{A}_1(t,i)$, we obtain
\begin{eqnarray*}\label{rrrle}
   &&\widehat{\mathfrak{A}}_{1}(t,\{Y\})=\mathfrak{V}_{1}(t,\{Y\}),\\
   &&\widehat{\mathfrak{A}}_{2}(t,\{Y\},s+1)=\mathfrak{V}_{2}(t,\{Y\},s+1)+
       \mathfrak{V}_{1}(t,\{Y\})\sum_{i_1=1}^s \widehat{\mathfrak{A}}_{2}(t,i_1,s+1),\\
   &&\widehat{\mathfrak{A}}_{3}(t,\{Y\},s+1,s+2)=\mathfrak{V}_{3}(t,\{Y\},s+1,s+2)+\\
   &&+2!\mathfrak{V}_{2}(t,\{Y\},s+1)\sum_{i_1=1}^{s+1}\widehat{\mathfrak{A}}_{2}(t,i_1,s+2)+\\
   &&+\mathfrak{V}_{1}(t,\{Y\})\big(\sum_{i_1=1}^s
       \widehat{\mathfrak{A}}_{3}(t,i_1,s+1,s+2)+2!\sum_{1=i_1<i_2}^s
       \widehat{\mathfrak{A}}_{2}(t,i_1,s+1)\widehat{\mathfrak{A}}_{2}(t,i_2,s+2)\big),
\end{eqnarray*}
where $\widehat{\mathfrak{A}}_{1+n}(t)$ is the $(1+n)th$-order $(n=0,1,2)$ scattering cumulant (\ref{scacu}).
It is easy to prove that solutions of these relations are given by the following expansions
\begin{eqnarray}\label{rrrls}
   &&\mathfrak{V}_{1}(t,\{Y\})=\widehat{\mathfrak{A}}_{1}(t,\{Y\})\doteq
       S_s(-t,1,\ldots,s)\mathfrak{I}_{s}(Y)\prod_{i=1}^{s}S_1(t,i),\\
   &&\mathfrak{V}_{2}(t,\{Y\},s+1)=\widehat{\mathfrak{A}}_{2}(t,\{Y\},s+1)-
       \widehat{\mathfrak{A}}_{1}(t,\{Y\})\sum_{i_1=1}^s
       \widehat{\mathfrak{A}}_{2}(t,i_1,s+1),\nonumber\\
   &&\mathfrak{V}_{3}(t,\{Y\},s+1,s+2)=\widehat{\mathfrak{A}}_{3}(t,\{Y\},s+1,s+2)
       -2!\widehat{\mathfrak{A}}_{2}(t,\{Y\},s+1)\times\nonumber\\
   &&\times\sum_{i_1=1}^{s+1}\widehat{\mathfrak{A}}_{2}(t,i_1,s+2)-
       \widehat{\mathfrak{A}}_{1}(t,\{Y\})\big(\sum_{i_1=1}^{s}
       \widehat{\mathfrak{A}}_{3}(t,i_1,s+1,s+2)-\nonumber\\
   &&-2!\sum_{i_1=1}^{s}\sum_{i_2=1}^{s+1}\widehat{\mathfrak{A}}_{2}(t,i_1,s+1)
       \widehat{\mathfrak{A}}_{2}(t,i_2,s+2)+
       2!\sum_{1=i_1< i_2}^{s}\widehat{\mathfrak{A}}_{2}(t,i_1,s+1)
       \widehat{\mathfrak{A}}_{2}(t,i_2,s+2)\big).\nonumber
\end{eqnarray}

In the general case solutions of recurrence relations (\ref{rrrl2}), i.e. the $(1+n)th$-order ($n\geq0$)
evolution operators $\mathfrak{V}_{1+n}(t)$, are defined by the expansions (\ref{skrr})
\begin{eqnarray}\label{skrrn}
    &&\hskip-8mm\mathfrak{V}_{1+n}(t,\{Y\},X\setminus Y)=n!\sum_{k=0}^{n}(-1)^k\,\sum_{m_1=1}^{n}\ldots
       \sum_{m_k=1}^{n-m_1-\ldots-m_{k-1}}\frac{1}{(n-m_1-\ldots-m_k)!}\times\\
    &&\hskip-8mm\times\widehat{\mathfrak{A}}_{1+n-m_1-\ldots-m_k}(t,\{Y\},s+1,
       \ldots,s+n-m_1-\ldots-m_k)\times\nonumber\\
    &&\hskip-8mm\times\prod_{j=1}^k\,\ \sum_{k_2^j=0}^{m_j}\ldots
       \sum_{k^j_{n-m_1-\ldots-m_j+s}=0}^{k^j_{n-m_1-\ldots-m_j+s-1}}\prod_{i_j=1}^{s+n-m_1-\ldots-m_j}
       \frac{1}{(k^j_{n-m_1-\ldots-m_j+s+1-i_j}-k^j_{n-m_1-\ldots-m_j+s+2-i_j})!}\times\nonumber\\
    &&\hskip-8mm\times\widehat{\mathfrak{A}}_{1+k^j_{n-m_1-\ldots-m_j+s+1-i_j}-k^j_{n-m_1-\ldots-m_j+s+2-i_j}}
       (t,i_{j},s+n-m_1-\ldots-m_j+1+\nonumber \\
    &&\hskip-8mm+k^j_{s+n-m_1-\ldots-m_j+2-i_j},\ldots,s+n-m_1-\ldots-m_j+k^j_{s+n-m_1-\ldots-m_j+1-i_j}),\nonumber\\ \nonumber\\
    &&\hskip-8mm s\geq2,\, n\geq0.\nonumber
\end{eqnarray}
This fact is verified as a result of the substitution of expressions (\ref{skrrn}) into recurrence
relations (\ref{rrrl2}). It should be emphasized that the correlations generated by the dynamics of
hard spheres and initial correlations connected with the forbidden configurations are completely
governed by evolution operators (\ref{skrrn}).

Let us indicate some properties of evolution operators (\ref{skrrn}). In the sense of the weak
convergence of the space $L^{1}(\mathbb{R}^{3s}\times\mathbb{R}^{3s}\backslash \mathbb{W}_s)$
the infinitesimal generator of first-order evolution operator (\ref{rrrls}) is defined by operator
(\ref{aLint})
\begin{eqnarray*}
    &&\lim\limits_{t\rightarrow 0}\frac{1}{t}(\mathfrak{V}_{1}(t,\{Y\})-I)f_{s}
       =\sum_{i<j=1}^s\mathcal{L}^*_{\mathrm{int}}(i,j)f_{s},
\end{eqnarray*}
where the operator $\mathcal{L}^*_{\mathrm{int}}(i,j)$ is defined by formula (\ref{aLint}).

In the general case, where $n\geq1$, in the same sense as for evolution operators (\ref{rrrls})
it holds that
\begin{eqnarray*}
    &&\lim\limits_{t\rightarrow 0}\frac{1}{t}\,\mathfrak{V}_{1+n}(t,\{Y\},Y\backslash X)f_{s+n}=0.
\end{eqnarray*}

Since in case of a system of non-interacting particles for scattering cumulants (\ref{scacu}) the
next equalities hold: $\widehat{\mathfrak{A}}_{1+n}(t,\{Y\},X \backslash Y)=\prod_{i=1}^s S_1(-t,i)
\mathfrak{I}_s(Y)\prod_{j=1}^s S_1(t,j)\delta_{n,0},\,n\geq0$, then we have
\begin{eqnarray*}
   &&\mathfrak{V}_{1+n}(t,\{Y\},X \backslash Y)=
      \prod_{i=1}^s S_1(-t,i)\mathfrak{I}_s(Y)\prod_{j=1}^s S_1(t,j)\delta_{n,0},
\end{eqnarray*}
where $\delta_{n,0}$ is a Kronecker symbol.

Similarly, at initial time $t=0$ it is true that:
$\mathfrak{V}_{1+n}(0,\{Y\},X \backslash Y)=\mathfrak{I}_s(Y)\delta_{n,0}$.

\subsection{The derivation of the generalized Enskog equation}
We establish that the one-particle distribution function defined by expansion
(\ref{F(t)}),(\ref{nLkymyl}) in case of $s=1$, i.e.
\begin{eqnarray}\label{F(t)1}
    &&\hskip-12mmF_{1}(t,x_1)=\sum\limits_{n=0}^{\infty}\frac{1}{n!}
      \int_{(\mathbb{R}^3\times\mathbb{R}^3)^n}dx_2\ldots dx_{n+1}\,
      \mathfrak{A}_{1+n}(-t,1,\ldots,n+1)\prod_{i=1}^{n+1}F_{1}^0(x_i)
      \mathcal{X}_{\mathbb{R}^{3(1+n)}\setminus \mathbb{W}_{1+n}},
\end{eqnarray}
is governed by kinetic equation (\ref{gke1}) or (\ref{gke1n}).

In view of the validity of the following equalities for cumulants (\ref{nLkymyl}) of groups
(\ref{Sspher}) for $f\in L^{1}_{0}$ in the sense of the norm convergence \cite{CGP97},\cite{Pe08}
\begin{eqnarray*}
   &&\lim\limits_{t\rightarrow 0}\frac{1}{t}\mathfrak{A}_{1}(-t,1)f_{1}(x_1)
     =\mathcal{L}^*_1(1)f_{1}(x_1),\\
   &&\lim\limits_{t\rightarrow 0}\frac{1}{t}\int_{\mathbb{R}^3\times\mathbb{R}^3}
     dx_2\mathfrak{A}_{2}(-t,1,2)f_{2}(x_1,x_2)
     =\int_{\mathbb{R}^3\times \mathbb{R}^3}dx_2\mathcal{L}^*_{\mathrm{int}}(1,2)f_{2}(x_1,x_2),\\
   &&\lim\limits_{t\rightarrow 0}\frac{1}{t}
     \int_{\mathbb{R}^{3n}\times(\mathbb{R}^{3n}\setminus \mathbb{W}_n)}
     dx_2\ldots dx_{n+1}\,\mathfrak{A}_{1+n}(-t,1,\ldots,n+1)f_{n+1}=0,\quad n\geq2,
\end{eqnarray*}
where the operators $\mathcal{L}^*_1(1)$ and $\mathcal{L}^*_{\mathrm{int}}(1,2)$ are defined by
formulas (\ref{OperL}) and (\ref{aLint}) respectively, as a result of the differentiation over
the time variable of expression (\ref{F(t)1}) in the sense of pointwise convergence on the space
$L^{1}(\mathbb{R}^3\times\mathbb{R}^3)$ we obtain \cite{CGP97},\cite{GP85}
\begin{eqnarray}\label{de}
  &&\hskip-7mm\frac{\partial}{\partial t}F_{1}(t,x_1)=
     -\langle p_1,\frac {\partial}{\partial q_1}\rangle F_{1}(t,x_1)+
     \int_{\mathbb{R}^3\times\mathbb{R}^3}dx_2\mathcal{L}^*_{\mathrm{int}}(1,2)
     \sum\limits_{n=0}^{\infty}\frac{1}{n!}\times\\
  &&\hskip-7mm\times\int_{(\mathbb{R}^{3}\times\mathbb{R}^{3})^{n}}
     dx_{3}\ldots dx_{n+2}\,\mathfrak{A}_{1+n}(-t,\{1,2\},3,\ldots,n+2)
     \prod_{i=1}^{n+2}F_1^0(x_i)\mathcal{X}_{\mathbb{R}^{3(n+2)}\setminus \mathbb{W}_{n+2}}.\nonumber
\end{eqnarray}
To represent the second term of the right-hand side in this equality in terms of one-particle
distribution function (\ref{F(t)1}) we expand cumulants (\ref{nLkymyl}) of groups (\ref{Sspher})
into kinetic cluster expansions (\ref{rrrl2}) for the case $s=2$. Then we transform the series over
the summation index $n$ and the sum over the index $k_1$ to the two-fold series. As a result it holds
\begin{eqnarray}\label{scint}
  &&\hskip-9mm\sum\limits_{n=0}^{\infty}\frac{1}{n!}\int_{(\mathbb{R}^{3}\times\mathbb{R}^{3})^{n}}
     dx_{3}\ldots dx_{n+2}\,\mathfrak{A}_{1+n}(-t,\{1,2\},3,\ldots,n+2)\mathfrak{I}_{n+2}(1,\ldots,n+2)
     \prod_{i=1}^{n+2}F_1^0(x_i)\\
  &&\hskip-9mm=\sum\limits_{n=0}^{\infty}\frac{1}{n!}\sum_{k_1=0}^{\infty}
     \int_{(\mathbb{R}^{3}\times\mathbb{R}^{3})^{n+k_1}}dx_{3}\ldots dx_{n+2+k_1}
     \mathfrak{V}_{1+n}(t,\{1,2\},3,\ldots,n+2)\sum_{k_2=0}^{k_1}\ldots\nonumber\\
  &&\hskip-9mm\ldots\sum_{k_{n+2}=0}^{k_{n+1}}\frac{1}{k_{n+2}!(k_{n+1}-k_{n+2})!
     \ldots(k_1-k_2)!}\prod_{i=1}^{n+2}{\mathfrak{A}}_{1+k_{n+3-i}-k_{n+4-i}}(-t,i,n+3+k_{n+4-i},\nonumber \\
  &&\hskip-9mm\ldots,n+2+k_{n+3-i})\mathcal{X}_{1+k_{n+3-i}-k_{n+4-i}}(q_i,q_{n+3+k_{n+4-i}},
     \ldots,q_{n+2+k_{n+3-i}})\prod_{j=1}^{n+2+k_1}F_1^0(x_j). \nonumber
\end{eqnarray}
The last series can be expressed in terms of one-particle distribution function (\ref{F(t)1})
according to the following product formula for series (\ref{F(t)1})
\begin{eqnarray}\label{Pr}
   &&\hskip-7mm\prod_{i=1}^{n+2}F_1(t,x_i)=\sum_{k_1=0}^{\infty}\sum_{k_2=0}^{k_1}\ldots
     \sum_{k_{n+2}=0}^{k_{n+1}}\frac{1}{k_{n+2}!(k_{n+1}-k_{n+2})!\ldots(k_1-k_2)!}\times \\
   &&\hskip-7mm\times\int_{(\mathbb{R}^{3}\times\mathbb{R}^{3})^{k_1}}dx_{n+3}\ldots dx_{n+2+k_1}
     \prod_{i=1}^{n+2}{\mathfrak{A}}_{1+k_{n+3-i}-k_{n+4-i}}(-t,i,n+3+k_{n+4-i},\nonumber \\
   &&\hskip-7mm\ldots,n+2+k_{n+3-i})\mathcal{X}_{1+k_{n+3-i}-k_{n+4-i}}(q_i,q_{n+3+k_{n+4-i}},
     \ldots,q_{n+2+k_{n+3-i}})\prod_{j=1}^{n+2+k_1}F_1^0(x_j).\nonumber
\end{eqnarray}

According to equalities (\ref{scint}) and (\ref{Pr}), and taking into account definition (\ref{aLint})
of the operator $\mathcal{L}^*_{\mathrm{int}}(1,2)$ for $t>0$, from equality (\ref{de}) we finally derive
\begin{eqnarray*}
   &&\frac{\partial}{\partial t}F_{1}(t,x_1)=
      -\langle p_1,\frac{\partial}{\partial q_1}\rangle F_{1}(t,x_1)+\\
   &&+\sigma^2\sum_{n=0}^{\infty}\frac{1}{n!}\int_{\mathbb{R}^3\times\mathbb{S}^2_+}d p_2 d\eta
      \int_{(\mathbb{R}^{3}\times\mathbb{R}^{3})^{n}}dx_3\ldots dx_{n+2}\langle\eta,(p_1-p_2)\rangle\times\\
   &&\times\big(\mathfrak{V}_{1+n}(t,\{1^{*},2^{*}_{-}\},3,\ldots,n+2)
      F_1(t,q_1,p_1^{*})F_1(t,q_1-\sigma\eta,p_2^{*})\prod_{i=3}^{n+2}F_{1}(t,x_i)-\\
   &&-\mathfrak{V}_{1+n}(t,\{1,2_{+}\},3,\ldots,n+2)F_1(t,x_1)
      F_1(t,q_1+\sigma\eta,p_2)\prod_{i=3}^{n+2}F_{1}(t,x_i)\big),
\end{eqnarray*}
where we used notations accepted in equation (\ref{gke1}). We treat the constructed identity for the
one-particle distribution function as the kinetic equation for a hard sphere system. We refer to this
evolution equation as the generalized Enskog kinetic equation. The collision integral series converges
(see next section in the general case) under the condition:
$\|F_1(t)\|_{L^1(\mathbb{R}^3\times\mathbb{R}^3)}<e^{-8}$.

We emphasize that coefficients of the derived kinetic equation (\ref{gke1}) are determined by initial
correlations connected with the forbidden configurations of particles.

For quantum many-particle system the generalized kinetic equation has been formulated in \cite{GT} (see
also a review \cite{GerUJP}) and for systems of classical particles interacting via smooth potential it
was stated by other method in \cite{CGP97},\cite{GP97,GP98} and for discrete velocity models in \cite{BGL}.


\section{Marginal functionals of the state}

\subsection{A mean value functional within the framework of the kinetic evolution}
By means of kinetic cluster expansions (\ref{rrrl2}) in case of $s\geq2$, solution expansions
(\ref{F(t)}) of the BBGKY hierarchy (\ref{NelBog1}) for hard spheres can be represented in the
form of the expansions with respect to one-particle distribution function (\ref{F(t)1}) which
is governed by the generalized Enskog equation (\ref{gke1}).

Indeed, we expand cumulants (\ref{nLkymyl}) of groups (\ref{Sspher}) from solution expansions
(\ref{F(t)}) into kinetic cluster expansions (\ref{rrrl2}) in case of $s\geq2$. Then, transforming
the series over the summation index $n$ and the sum over the index $k_1$ to the two-fold series,
we obtain
\begin{eqnarray*}
   &&F_{s}(t,x_1,\ldots,x_s)=\sum\limits_{n=0}^{\infty}\frac{1}{n!}\int\limits_{(\mathbb{R}^{3}\times\mathbb{R}^{3})^{n}}
      dx_{s+1}\ldots dx_{s+n}\mathfrak{A}_{1+n}(-t,\{Y\},X\setminus Y)\mathfrak{I}_{s+n}(X)
      \prod_{i=1}^{s+n}F_{1}^0(x_i)=\\
   &&=\sum\limits_{n=0}^{\infty}\frac{1}{n!}\int_{(\mathbb{R}^{3}\times\mathbb{R}^{3})^n}dx_{s+1}\ldots
      dx_{s+n}\,\mathfrak{V}_{1+n}(t,\{Y\},X\setminus Y)\sum_{k_1=0}^{\infty}\sum_{k_2=0}^{k_1}\ldots\\
   &&\ldots\sum_{k_{n+s}=0}^{k_{n+s-1}}\frac{1}{k_{n+s}!(k_{n+s-1}-k_{n+s})!\ldots(k_1-k_2)!}
      \int_{(\mathbb{R}^{3}\times\mathbb{R}^{3})^{k_1}}dx_{n+s+1}\ldots\\
   &&\ldots dx_{n+s+k_1}\prod_{i=1}^{n+s}{\mathfrak{A}}_{1+k_{n+s+1-i}-k_{n+s+2-i}}(-t,i,n+s+1+k_{n+s+2-i},
      \ldots,n+s+k_{n+s+1-i})\\
   &&\times\prod_{j=1}^{n+s+k_1}F_1^0(x_j)
      \mathcal{X}_{1+k_{n+s+1-i}-k_{n+s+2-i}}(q_i,q_{n+s+1+k_{n+s+2-i}},\ldots,q_{n+s+k_{n+s+1-i}}),
\end{eqnarray*}
where we use notations accepted above.
According to the validity of product formula (\ref{Pr}) in case of the arbitrary index $s\geq2$,
in obtained expansion the series over the index $k_1$ can be expressed in terms of one-particle
distribution function (\ref{F(t)1}). Hence the following equality holds
\begin{eqnarray*}
   &&\sum\limits_{n=0}^{\infty}\frac{1}{n!}\int_{(\mathbb{R}^{3}\times\mathbb{R}^{3})^{n}}dx_{s+1}\ldots dx_{s+n}
       \mathfrak{A}_{1+n}(-t,\{Y\},X\setminus Y)\mathfrak{I}_{s+n}(X)\prod_{i=1}^{s+n}F_1^0(x_i)=\\
   &&=\sum\limits_{n=0}^{\infty}\frac{1}{n!}\int_{(\mathbb{R}^{3}\times\mathbb{R}^{3})^{n}}dx_{s+1}\ldots dx_{s+n}
       \mathfrak{V}_{1+n}(t,\{Y\},X\setminus Y)\prod_{i=1}^{s+n}F_{1}(t,x_i),
\end{eqnarray*}
where the $(n+1)th$-order evolution operator $\mathfrak{V}_{1+n}(t)$ is defined by formula (\ref{skrr}).

Thus, if initial data is completely defined by a one-particle distribution function on allowed
configurations, then the evolution of states governed by the BBGKY hierarchy (\ref{NelBog1})
can be completely described by the generalized kinetic equation (\ref{gke1}) and the sequence
of marginal functionals of the state $F_{s}\big(t,x_{1},\ldots,x_{s}\mid F_{1}(t)\big),\,s\geq2$,
defined by expansions (\ref{f}). In other words, for mentioned initial states all possible states
of a hard sphere system at arbitrary moment of time can be described within the framework of a
one-particle distribution function without any approximations. Therefore the statement of
section 2.2 about an equivalence of stated Cauchy problems is valid.

Typical properties for the kinetic description of the evolution of constructed marginal functionals
of the state (\ref{f}) are induced by the properties of evolution operators (\ref{skrr}).

We indicate that expansions (\ref{f}) of marginal functionals of the state are nonequilibrium analog
of the Mayer-Ursell expansions over powers of the density of equilibrium distribution functions.
In particular, if the one-particle distribution function is the Maxwellian distribution function,
then marginal functionals of the state (\ref{f}) is transformed into the Mayer-Ursell expansions
\cite{B46},\cite{UF,GP63}.

Owing that a sequence of solutions of the BBGKY hierarchy is represented in terms of a solution of
the generalized Enskog equation by the marginal functionals of the state (\ref{f}), the average
value of the $s$-ary marginal observable $B^{(s)}=(0,\ldots,0,a_{s}(x_1,\ldots,x_s),0,\ldots)$ \cite{GerUJP}
is defined in terms of the functional $F_s(t,x_1,\ldots,x_s|F_1(t))$, by the formula
\begin{eqnarray*}
  &&\langle B^{(s)}\rangle(t)=\frac{1}{s!}\int_{(\mathbb{R}^{3}\times\mathbb{R}^{3})^{s}}
     dx_1\ldots dx_s a_s(x_1,\ldots,x_s)F_s(t,x_1,\ldots,x_s\mid F_1(t)),
\end{eqnarray*}
and, in particular the average value of the additive marginal observable $B^{(1)}=(0,a_{1}(x_1),0,\ldots)$
is determined by
\begin{eqnarray}\label{averageg}
  &&\langle B^{(1)}\rangle(t)=\int_{\mathbb{R}^{3}\times\mathbb{R}^{3}}dx_1 a_1(x_1)F_1(t,x_1),
\end{eqnarray}
where the function $F_1(t,x_1)$ is a solution of the Cauchy problem (\ref{gke1})-(\ref{2}).

In fact functionals (\ref{f}) characterize the correlations of states of a hard sphere system.
We illustrate this statement by giving an example of the marginal functional from the collision
integral of the generalized Enskog equation. In this case the marginal correlation functional
$G_{2}\big(t,x_1,x_2\mid F_{1}(t)\big)$ is defined by the cluster expansion
\begin{eqnarray*}
    &&F_{2}\big(t,x_1,x_2\mid F_{1}(t)\big)=F_{1}(t,x_1)F_{1}(t,x_2)+G_{2}\big(t,x_1,x_2\mid F_{1}(t)\big),
\end{eqnarray*}
and it is represented by the series similar to (\ref{f})
\begin{eqnarray}\label{cf}
   &&G_{2}\big(t,x_1,x_2\mid F_{1}(t)\big)=(\widehat{\mathfrak{A}}_{1}(t,\{1,2\})-I)F_{1}(t,x_1)F_{1}(t,x_2)+\\
   &&+\sum\limits_{n=1}^{\infty}\frac{1}{n!}\int_{(\mathbb{R}^{3}\times\mathbb{R}^{3})^{n}}dx_{3}\ldots dx_{n+2}
      \,\mathfrak{V}_{1+n}(t,\{1,2\},3,\ldots,n+2)\prod_{i=1}^{n+2}F_{1}(t,x_i).\nonumber
\end{eqnarray}
Basically this cluster expansion gives the classification of all possible currently in use scaling
limits \cite{Sp91},\cite{GH,Sp80} since it characterizes the deviation of the state from uncorrelated
state, i.e. the state satisfying a chaos property. In the scaling limits it is assumed that chaos
property of initial state preserves in time, and hence the scaling limit means such limit of dimensionless
parameters of a system in which the marginal correlation functional $G_{2}\big(t,1,2\mid F_{1}(t)\big)$
vanishes. According to definition (\ref{cf}), it is possible, if particles of every finite particle
cluster move without collisions. In consequence of the definition of a dispersion, for example, of an
additive-type observable
\begin{eqnarray*}
  &&\hskip-5mm\langle(B^{(1)}-\langle B^{(1)}\rangle(t))^2\rangle(t)=\\
  &&\hskip-5mm=\int_{\mathbb{R}^{3}\times\mathbb{R}^{3}}dx_1(a_1^2(x_1)-\langle B^{(1)}\rangle^2(t))F_{1}(t,x_1)+
     \int_{({\mathbb{R}^{3}\times\mathbb{R}^{3}})^2}dx_1dx_2 a_{1}(x_1)a_{1}(x_2)G_{2}\big(t,x_1,x_2\mid F_{1}(t)\big),
\end{eqnarray*}
where the functional $\langle B^{(1)}\rangle(t)$ is determined by expression (\ref{averageg}), the value of
the dispersion is minimal for states characterized by marginal correlation functionals (\ref{cf}) equal to
zero, i.e. from macroscopic point of view the scaling evolution of hard sphere states is the evolution with
minimal dispersion.

Thus, in case of the Boltzmann-Grad scaling limit \cite{CGP97,CIP,Sp91,Pe08,L,GH} the limit of solution
(\ref{F(t)1}) of the generalized Enskog equation is governed by the Boltzmann equation for hard spheres
and the limit of marginal functionals (\ref{f}) are products of the limit one-particle distribution functions
that means the propagation of a chaos.

\subsection{On the convergence of a series for marginal functionals}
We prove the existence of marginal functionals of the state (\ref{f}) for
$F_1(t)\in L^1(\mathbb{R}^3\times\mathbb{R}^3)$. In case of smooth interaction potentials
this fact was established in \cite{GP97,GP98} by other method.

Owing to the fact that operator (\ref{Sspher}) is  the isometric operator and the validity for
$|X\backslash Y|>1$ the identity:
$\sum_{\texttt{P}:\,(X\setminus Y)={\bigcup\limits}_i X_i}(-1)^{|\texttt{P}|-1}(|\texttt{P}|-1)!=0$,
for cumulants (\ref{nLkymyl}) of groups (\ref{Sspher}) the following identity holds
\begin{eqnarray*}
  &&\int_{(\mathbb{R}^{3}\times\mathbb{R}^{3})^{n}}dx_{s+1}\ldots dx_{s+n}\,
     \mathfrak{A}_{n}(t,X\setminus Y)\prod_{i=1}^{s+n}F_{1}(t,x_i)=0, \quad n>1.
\end{eqnarray*}
As a result of the validity of these identities we can represent marginal functionals of the state
(\ref{f}) by the following renormalized series
\begin{eqnarray}\label{ff}
  &&\hskip-12mm F_s(t,x_1,\ldots,x_s\mid F_1(t))=\sum_{n=0}^{\infty}\frac{1}{n!}
     \int_{(\mathbb{R}^{3}\times\mathbb{R}^{3})^{n}}dx_{s+1}\ldots dx_{s+n}
     \widetilde{\mathfrak{V}}_{1+n}\big(t,\{Y\},X\setminus Y\big)\prod_{i=1}^{s+n}F_{1}(t,x_i),
 \end{eqnarray}
where the $(1+n)th$-order renormalized evolution operator $\widetilde{\mathfrak{V}}_{1+n}(t)$ is defined
by the expansion
\begin{eqnarray}\label{skrr1}
  &&\hskip-7mm\widetilde{\mathfrak{V}}_{1+n}(t,\{Y\},X\setminus Y)=n!\sum_{k=0}^{n}(-1)^k\sum_{m_1=1}^{n}\ldots
     \sum_{m_k=1}^{n-m_1-\ldots-m_{k-1}}\frac{1}{(n-m_1-\ldots-m_k)!}\times\\
  &&\hskip-7mm\times\widehat{\mathfrak{A}}_{1+n-m_1-\ldots-m_k}(t,\{Y\},s+1,\ldots,s+n-m_1-\ldots-m_k)\times \nonumber\\
  &&\hskip-7mm\times\prod_{j=1}^k\sum_{k_2^j=0}^{m_j}\ldots\sum_{k^j_{s}=0}^{k^j_{s-1}}
     \prod_{i_j=1}^{s}\frac{1}{(k^j_{s+1-i_j}-k^j_{s+2-i_j})!}
     \widehat{\mathfrak{A}}_{1+k^j_{s+1-i_j}-k^j_{s+2-i_j}}(t,i_{j},\nonumber \\
  &&\hskip-7mm s+n-m_1-\ldots-m_j+1+k^j_{s+2-i_j},\ldots,s+n-m_1-\ldots-m_j+k^j_{s+1-i_j}).\nonumber
\end{eqnarray}
and as above in (\ref{skrr}) the next conventions are assumed: $k^j_1\equiv m_j,\, k^j_{s+1}\equiv 0$.

To prove the convergence of the series of marginal functional (\ref{f}) we use the estimate
of the cumulants $\mathfrak{A}_{1+n}(t)$ of groups of operators (\ref{Sspher}) \cite{GRS04}
\begin{eqnarray*}\label{est}
  &&\hskip-9mm\int_{(\mathbb{R}^{3}\times\mathbb{R}^{3})^{s+n}}dx_{1}\ldots dx_{s+n}
      \big|\mathfrak{A}_{1+n}(-t,\{Y\},X\setminus Y)\prod_{i=1}^{s+n}F_1(t,x_i)\big|\leq
      n!e^{n+2}\|F_1(t)\|_{L^1_1}^{s+n}.
\end{eqnarray*}
In consequence of this estimate, for $(1+n)th$-order renormalized evolution operator (\ref{skrr1})
the inequality holds
\begin{eqnarray}\label{EstNormOp}
   &&\hskip-9mm\int_{(\mathbb{R}^{3}\times\mathbb{R}^{3})^{s+n}}dx_{1}\ldots dx_{s+n}
      \big|\widetilde{\mathfrak{V}}_{1+n}(t,\{Y\},X\setminus Y)\prod_{i=1}^{s+n}F_1(t,x_i)\big|
      \leq n!\|F_1(t)\|_{L_1^1}^{s+n}\times\\
   &&\hskip-9mm\times\sum_{k=0}^{n}\sum_{m_1=1}^{n}\ldots\sum_{m_k=1}^{n-m_1-\ldots-m_{k-1}}
      e^{n-m_1-\ldots-m_k+2}\prod_{j=1}^k\sum_{k_2^j=0}^{m_j}\ldots\sum_{k^j_{s}=0}^{k_{s-1}^j}
      \prod_{i_j=1}^{s}e^{k^j_{s+1-i_j}-k^j_{s+2-i_j}+2}.\nonumber
\end{eqnarray}

So far as it is redefined that $k^j_1\equiv m_j,\,k^j_{s+1}\equiv 0$, we have
\begin{eqnarray}\label{equal1}
   &&\prod_{i_j=1}^{s}e^{k^j_{s+1-i_j}-k^j_{s+2-i_j}+2}=e^{m_j+2s}.
\end{eqnarray}

Then the following equality is true
\begin{eqnarray}\label{equal2}
   &&\sum_{k_2^j=0}^{m_j}\ldots \sum_{k^j_{s}=0}^{k_{s-1}^j}1=
     \frac{1}{(s-1)!}(m_j+1)(m_j+2)\ldots(m_j+s-1).
\end{eqnarray}
Indeed, using the equality (in case of $m=0$ it is the sum of the arithmetical progression)
\begin{eqnarray}\label{RecEqual}
   &&\sum_{k=1}^n\,k(k+1)\ldots(k+m)=\frac{1}{m+2}n(n+1)\ldots(n+m+1),
\end{eqnarray}
as a result of the sequential computation of sums we obtain
\begin{eqnarray*}
   &&\sum_{k_2^j=0}^{m_j}\ldots\sum_{k^j_{s}=0}^{k_{s-1}^j}1=
      \sum_{k_2^j=0}^{m_j}\ldots\sum_{k^j_{s-1}=0}^{k_{s-2}^j}(k_{s-1}^j+1)=
      \sum_{k_2^j=0}^{m_j}\ldots\sum_{k^j_{s-2}=0}^{k_{s-3}^j}\frac{(k_{s-2}^j+1)(k_{s-2}^j+2)}{2}= \\
   &&=\ldots=\frac{1}{(s-1)!}(m_j+1)(m_j+2)\ldots(m_j+s-1).
\end{eqnarray*}

According to estimates (\ref{equal1}) and (\ref{equal2}), we have
\begin{eqnarray}\label{inequal1}
   &&\hskip-9mm\prod_{j=1}^k\sum_{k_2^j=0}^{m_j}\ldots\sum_{k^j_{s}=0}^{k_{s-1}^j}\,
      \prod_{i_j=1}^{s}e^{k^j_{s+1-i_j}-k^j_{s+2-i_j}+2}=
      \prod_{j=1}^k\,e^{m_j+2s}\,\frac{(m_j+1)(m_j+2)\ldots(m_j+s-1)}{(s-1)!}\leq \\
   &&\hskip-9mm\leq\prod_{j=1}^k e^{m_j+2s}\frac{(m_j+s-1)^{s-1}}{(s-1)!}\leq \prod_{j=1}^k\,e^{2m_j+3s-1}. \nonumber
\end{eqnarray}

Thus, in consequence of estimate (\ref{inequal1}) inequality (\ref{EstNormOp}) takes the form
\begin{eqnarray}\label{EstNormOp1}
   &&\hskip-9mm\int_{(\mathbb{R}^{3}\times\mathbb{R}^{3})^{s+n}}dx_{1}\ldots dx_{s+n}
      \big|\widetilde{\mathfrak{V}}_{1+n}(t,\{Y\},X\setminus Y)\prod_{i=1}^{s+n}F_1(t,x_i)\big|\leq \\
   &&\hskip-9mm\leq n!\|F_1(t)\|_{L_1^1}^{s+n}e^{n+2}\sum_{k=0}^{n}e^{k(3s-1)}
      \sum_{m_1=1}^{n}\ldots\sum_{m_k=1}^{n-m_1-\ldots-m_{k-1}}e^{m_1+\ldots+m_k}.\nonumber
\end{eqnarray}

Let us estimate the sums $\sum_{m_1=1}^{n}\ldots\sum_{m_k=1}^{n-m_1-\ldots-m_{k-1}}e^{m_1+\ldots+m_k}$.
Calculating the sum over index $m_k$ as the sum of a geometric progression, we obtain
\begin{eqnarray*}
   &&\sum_{m_1=1}^{n}\ldots\sum_{m_k=1}^{n-m_1-\ldots-m_{k-1}}e^{m_1+\ldots+m_k}\leq
     e^{n+1}\sum_{m_1=1}^{n}\ldots\sum_{m_{k-1}=1}^{n-m_1-\ldots-m_{k-2}}1.
\end{eqnarray*}
Then, applying equality (\ref{RecEqual}), as a result of the sequential computation of sums the
following inequality holds
\begin{eqnarray*}\label{est1}
   &&\sum_{m_1=1}^{n}\ldots\sum_{m_k=1}^{n-m_1-\ldots-m_{k-1}}e^{m_1+\ldots+m_k}\leq
     e^{n+1}\sum_{m_1=1}^{n}\ldots\sum_{m_{k-2}=1}^{n-m_1-\ldots-m_{k-3}}(n-m_1-\ldots-m_{k-2})=\\
   &&=\ldots=e^{n+1}\frac{(n-k+2)\ldots(n-1)n}{(k-1)!}.
\end{eqnarray*}
According to the last inequality, for the majorant of estimate (\ref{EstNormOp1}) we get
\begin{eqnarray}\label{est12}
   &&\hskip-9mm\sum_{k=0}^ne^{k(3s-1)}\sum_{m_1=1}^{n}
      \ldots\sum_{m_k=1}^{n-m_1-\ldots-m_{k-1}}e^{m_1+\ldots+m_k}\leq \\
   &&\hskip-9mm\leq 1+e^{n+3s}+ e^{n+1}\sum_{k=2}^ne^{k(3s-1)}\frac{(n-k+2)\ldots(n-1)n}{(k-1)!}\leq\nonumber \\
   &&\hskip-9mm\leq1+e^{n+3s}+e^{n+1}\sum_{k=2}^ne^{k(3s-1)}\frac{n^{k-1}}{(k-1)!}
      \leq 1+e^{2n+1}\sum_{k=1}^n e^{k(3s-1)}.\nonumber
\end{eqnarray}

Thus, in consequence of inequalities (\ref{EstNormOp1}) and (\ref{est12}) for marginal functionals of the
state (\ref{ff}) the following estimate holds
\begin{eqnarray}\label{inequalF}
   &&\hskip-9mm\big\|F_s(t\mid F_1(t))\big\|_{L^1_{s}}\leq
      \sum_{n=0}^{\infty}\big\|F_1(t)\big\|_{L_1^1}^{s+n}e^{n+2}+
      \sum_{n=1}^{\infty}\big\|F_1(t)\big\|_{L_1^1}^{s+n}e^{3n+3}\sum_{k=1}^n e^{k(3s-1)}\leq \\
   &&\leq \big\|F_1(t)\big\|_{L_1^1}^{s}e^{2}\sum_{n=0}^{\infty}\big\|F_1(t)\big\|_{L_1^1}^{n}e^{n}+
      \big\|F_1(t)\big\|_{L_1^1}^{s}e^{3s+2}\sum_{n=1}^{\infty}\big\|F_1(t)\big\|_{L_1^1}^{n}e^{(3s+2)n}. \nonumber
\end{eqnarray}
Hence, the sequence of marginal functionals of the state $F_s(t,x_1,\ldots,x_s\mid F_1(t)),\,s\geq2$, exists and
represents by converged series (\ref{f}) provided that
\begin{eqnarray}\label{estFunc1}
  &&\|F_1(t)\|_{L^{1}(\mathbb{R}^3\times\mathbb{R}^3)}<e^{-(3s+2)}.
\end{eqnarray}

On the other hand, owing the estimate for series (\ref{F(t)1})
\begin{eqnarray}\label{estFunc12}
   &&\|F_1(t)\|_{L^{1}(\mathbb{R}^3\times\mathbb{R}^3)}
      \leq \|F_1^0\|_{L^{1}(\mathbb{R}^3\times\mathbb{R}^3)}e^{2}\sum_{n=0}^{\infty}e^n
      \|F_1^0\|_{L^{1}(\mathbb{R}^3\times\mathbb{R}^3)}^n,
\end{eqnarray}
and according to estimates (\ref{estFunc1}) and (\ref{estFunc12}), we conclude that the Enskog collision
integral exists under the following condition on initial data
\begin{eqnarray}\label{estCon}
   &&\|F_1^0\|_{L^{1}(\mathbb{R}^3\times\mathbb{R}^3)}<\frac{e^{-10}}{1+e^{-9}}.
\end{eqnarray}


\section{Some properties of the collision integral}

\subsection{The Markovian generalized Enskog equation}
To state the Markovian generalized Enskog equation we first represent the collision integral $\mathcal{I}_{GEE}$
of the generalized Enskog equation (\ref{gke1}) as an expansion with respect to the Boltzman-Enskog collision
integral (in case of $t>0$)
\begin{eqnarray*}
  &&\hskip-7mm\mathcal{I}_{BEE}\equiv\sigma^2\int_{\mathbb{R}^3\times\mathbb{S}^2_+}d p_2 d\eta
     \langle\eta,(p_1-p_2)\rangle
     \big(F_1(t,q_1,p^*_1)F_1(t,q_1-\sigma\eta,p^*_2)-F_1(t,x_1)F_1(q_1+\sigma\eta,p_2)\big),\nonumber
\end{eqnarray*}
where the momenta $p_{1}^{*},p_{2}^{*}$ are defined by equalities (\ref{eq:momenta}). We observe that such
expansion of the collision integral $\mathcal{I}_{GEE}$ is given in terms of marginal correlation functionals
(\ref{cf}). Indeed, in view of that it holds: $(\widehat{\mathfrak{A}}_{1}(t,\{1^{\sharp},2^{\sharp}_{\pm}\})-I)
F_{1}(t,q_1,p_1^{\sharp})F_{1}(t,q_1\pm\sigma\eta,p_2^{\sharp})=0$, we have
\begin{eqnarray}\label{GEE}
  &&\mathcal{I}_{GEE}=\mathcal{I}_{BEE}+\sum_{n=1}^{\infty}\mathcal{I}^{(n)}_{GEE}\equiv\\
  &&\equiv\mathcal{I}_{BEE}+\sigma^2\sum_{n=1}^{\infty}\frac{1}{n!}
     \int_{\mathbb{R}^3\times\mathbb{S}^2_+}d p_2 d\eta
     \int_{(\mathbb{R}^{3}\times\mathbb{R}^{3})^{n}}dx_3\ldots dx_{n+2}\langle\eta,(p_1-p_2)\rangle\times\nonumber\\
  &&\times\big(\mathfrak{V}_{1+n}(t,\{1^{*},2^{*}_{-}\},3,\ldots,n+2)
     F_1(t,q_1,p_1^{*})F_1(t,q_1-\sigma\eta,p_2^{*})\prod_{i=3}^{n+2}F_{1}(t,x_i)-\nonumber\\
  &&-\mathfrak{V}_{1+n}(t,\{1,2_{+}\},3,\ldots,n+2)F_1(t,x_1)
     F_1(t,q_1+\sigma\eta,p_2)\prod_{i=3}^{n+2}F_{1}(t,x_i)\big)\nonumber,
\end{eqnarray}
where notations accepted in equation (\ref{gke1}) are used.

On the kinetic (macroscopic) scale the typical length for the kinetic phenomena is the mean free path.
Then, observing that in the kinetic scale of the variation of variables \cite{CIP},\cite{Sp80} the
groups of operators (\ref{Sspher}) of finitely many hard spheres describe a fast evolutionary processes
and hence they depend on microscopic time variable $\varepsilon^{-1}t$, where $\varepsilon\geq0$ is a scale
parameter (the ratio of the collision time to the mean free time), the dimensionless marginal functionals
of the state are represented in the form: $F_{s}\big(\varepsilon^{-1}t,x_1,\ldots,x_s\mid F_{1}(t)\big),\,s\geq2$.
In the formal limit $\varepsilon\rightarrow0$, (the Markovian approximation) the limit marginal functional of the state
$F_{s}(x_1,\ldots,x_s\mid F_{1}(t))$ is represented by expansion (\ref{f}) (and also (\ref{cf})) with the limit
generating evolution operators $\lim_{\varepsilon\rightarrow0}\mathfrak{V}_{1+n}(\varepsilon^{-1}t),\,n\geq0$,
for example,
\begin{eqnarray*}
   &&\lim\limits_{\varepsilon\rightarrow0}\mathfrak{V}_{1}(\varepsilon^{-1}t,\{Y\})=
      \lim\limits_{\varepsilon\rightarrow0}\widehat{\mathfrak{A}}_{1}(\varepsilon^{-1}t,\{Y\}).
\end{eqnarray*}

We note that the Markovian limit of the first two terms of expansions (\ref{f}) of dimensionless marginal
functionals of the state coincide with corresponding terms of particular solution of the BBGKY hierarchy
constructed by N.N. Bogolyubov \cite{CGP97},\cite{B46},\cite{UF} on basis of the perturbation method with
the use of the condition of weakening of correlation.

In the formal Markovian limit $\varepsilon \rightarrow 0$ from the generalized Enskog equation (\ref{gke1})
we derive the Markovian generalized Enskog equation
\begin{eqnarray}\label{MarkGke1}
  &&\hskip-7mm\frac{\partial}{\partial t}F_{1}(t,x_1)=
    -\langle p_1,\frac{\partial}{\partial q_1}\rangle F_{1}(t,x_1)+\mathcal{I}_{BEE}+\\
  &&\hskip-7mm+\sigma^2\int_{\mathbb{R}^3\times\mathbb{S}^2_+}d p_2 d\eta
     \langle\eta,(p_1-p_2)\rangle\sum_{n=1}^{\infty}\frac{1}{n!}
     \int_{(\mathbb{R}^{3}\times\mathbb{R}^{3})^{n}}dx_3\ldots dx_{n+2}
     \big(\mathfrak{V}_{1+n}(\{1^{*},2^{*}_{-}\},3,\ldots\nonumber\\
  &&\hskip-7mm\ldots,n+2)F_1(t,q_1,p_1^{*})F_1(t,q_1-\sigma\eta,p_2^{*})\prod_{i=3}^{n+2}F_{1}(t,x_i)-\nonumber\\
  &&\hskip-7mm-\mathfrak{V}_{1+n}(\{1,2_{+}\},3,\ldots,n+2)F_1(t,x_1)
     F_1(t,q_1+\sigma\eta,p_2)\prod_{i=3}^{n+2}F_{1}(t,x_i)\big),\nonumber
\end{eqnarray}
where by $\mathfrak{V}_{1+n}(\{1^{\sharp},2^{\sharp}_{\pm}\},3,\ldots,n+2)$ we denote as above the limit generating
evolution operators: $\mathfrak{V}_{1+n}=\lim_{\varepsilon\rightarrow0}\mathfrak{V}_{1+n}(\varepsilon^{-1}t),\,n\geq0$.

Thus, in the kinetic scale of the variation of variables the structure of the collision integral of the generalized
Enskog equation (\ref{MarkGke1}) describes all possible correlations which are created by hard sphere dynamics and
by the propagation of initial correlations stipulated by forbidden configurations.

We remark that for the first time two terms of the Markovian collision integral in case of smooth interaction
potentials was established by N.N. Bogolyubov \cite{B46} within the framework of the perturbation theory and
using an analogy with the virial equilibrium expansions, these terms of the collision integral in non-perturbative
form were formulated by M.S. Green and R.A. Piccirelli \cite{GP63},\cite{Gre56} and by E.G.D. Cohen \cite{C}.

To derive the Bogolyubov collision integral for hard spheres from (\ref{MarkGke1}) we apply analogs of the
Duhamel equations \cite{BA} to generating operators of marginal correlation functionals (\ref{cf}). Indeed,
if $f_3\in L^1(\mathbb{R}^{9}\times(\mathbb{R}^{9}\setminus \mathbb{W}_3))$, for the scattering cumulant
$\widehat{\mathfrak{A}}_{2}(\varepsilon^{-1}t,{1,2},3)$ (\ref{scacu}) an analog of the Duhamel equation
holds \cite{Pe08}
\begin{eqnarray*}\label{Duam}
  &&\widehat{\mathfrak{A}}_{2}(\varepsilon^{-1}t,\{1,2\},3)f_3(x_1,x_2,x_3)=
      \int_0^{\varepsilon^{-1}t}d\tau S_2(-\tau,1,2)S_1(-\tau,3)
      \sum_{i_1=1}^2\mathcal{L}^{*}_{\mathrm{int}}(i_1,3)\times\\
  &&\times S_3(-\varepsilon^{-1}t+\tau,1,2,3)\mathcal{X}_{\mathbb{R}^{9}\setminus \mathbb{W}_{3}}
      \prod_{i_2=1}^{3}S_1(\varepsilon^{-1}t,i_2)f_3(x_1,x_2,x_3),\nonumber
\end{eqnarray*}
and, consequently, for the first term of the collision integral expansion we obtain
\begin{eqnarray*}
   &&\mathcal{I}^{(1)}_{GEE}=\mathcal{I}^{(1)}_{GEE,+}-\mathcal{I}^{(1)}_{GEE,-},
\end{eqnarray*}
where we denote by $\mathcal{I}^{(1)}_{GEE,\pm}$ the following parts of the collision integral
$\mathcal{I}^{(1)}_{GEE}$
\begin{eqnarray*}\label{BogEq}
   &&\hskip-7mm\mathcal{I}^{(1)}_{GEE,\pm}=\sigma^2\int_{\mathbb{R}^3\times\mathbb{S}^2_+}d p_2d\eta
      \langle\eta,(p_1-p_2)\rangle\int_{\mathbb{R}^3\times\mathbb{R}^3}dx_3
      \int_0^{\varepsilon^{-1}t}d\tau S_2(-\tau,1^{\sharp},2^{\sharp}_{\pm})
      \big((\mathcal{L}^{*}_{\mathrm{int}}(1^{\sharp},3)+\\
   &&\hskip-7mm+\mathcal{L}^{*}_{\mathrm{int}}(2^{\sharp}_{\pm},3))
     S_3(-\varepsilon^{-1}t+\tau,1^{\sharp},2^{\sharp}_{\pm},3)
      \mathcal{X}_2(q_1,q_3)\mathcal{X}_2(q_1\pm\sigma\eta,q_3)\times\\
   &&\hskip-7mm \times S_1(\varepsilon^{-1}t-\tau,1^{\sharp})
      S_1(\varepsilon^{-1}t-\tau,2^{\sharp}_{\pm})S_1(\varepsilon^{-1}t-\tau,3)-S_2(-\varepsilon^{-1}t+\tau,1^{\sharp},2^{\sharp}_{\pm})
      S_1(\varepsilon^{-1}t-\tau,1^{\sharp})\times\\
   &&\hskip-7mm\times S_1(\varepsilon^{-1}t-\tau,2^{\sharp}_{\pm})
     (\mathcal{L}^{*}_{\mathrm{int}}(1^{\sharp},3)S_2(-\varepsilon^{-1}t+\tau,1^{\sharp},3)
      \mathcal{X}_2(q_1,q_3)S_1(\varepsilon^{-1}t-\tau,1^{\sharp})S_1(\varepsilon^{-1}t-\tau,3) \\
   &&\hskip-7mm+\mathcal{L}^{*}_{\mathrm{int}}(2^{\sharp}_{\pm},3)
      S_2(-\varepsilon^{-1}t+\tau,2^{\sharp}_{\pm},3)\mathcal{X}_2(q_1\pm\sigma\eta,q_3)
      S_1(\varepsilon^{-1}t-\tau,2^{\sharp}_{\pm})\times\\
   &&\hskip-7mm\times S_1(\varepsilon^{-1}t-\tau,3))\big)S_1(\tau,1^{\sharp})S_1(\tau,2^{\sharp}_{\pm})S_1(\tau,3)
      F_1(t,q_1,p_1^{\sharp})F_1(t,q_1\pm\sigma\eta,p_2^{\sharp})F_1(t,x_3).\nonumber
\end{eqnarray*}
In the Markovian limit $\varepsilon\rightarrow0$ this expression coincides with the corresponding correction
to the Boltzmann-Enskog collision integral constructed by the perturbation method with the use of the
weakening of correlation condition in \cite{B46},\cite{UF}.

\subsection{To the classification of the Enskog-type kinetic equations}
We establish links of the Markovian generalized Enskog equation (\ref{MarkGke1}) with the revised
Enskog equation \cite{BLPT,Pol,Pol1}. The classification and links of the Enskog-type kinetic equations
derived upon some heuristic arguments is given in \cite{BLPT},\cite{BL}.

The collision integral of the revised Enskog equation has the form \cite{BLPT,Pol,Pol1}
\begin{eqnarray*}
  &&\hskip-8mm\mathcal{I}_{REE}=\sigma^2\int_{\mathbb{R}^3\times \mathbb{S}^2_+}d p_2 d\eta
    \langle\eta,(p_1-p_2)\rangle \big(g_2(q_1,q_1-\sigma\eta\mid F_1(t))F_1(t,q_1,p_1^*)F_1(t,q_1-\sigma\eta,p_2^*)-\\
  &&\hskip-8mmg_2(q_1,q_1+\sigma\eta\mid F_1(t))F_1(t,q_1,p_1)F_1(t,q_1+\eta\sigma,p_2)\big),\nonumber
\end{eqnarray*}
where the functional $g_2(q_1,q_2\mid F_1(t))$ is defined in terms of the formal Mayer cluster expansion
\begin{eqnarray}\label{ClustExp}
  &&g_2(q_1,q_2\mid F_1(t))=\\
  &&=\mathcal{X}_2(q_1,q_2)\big(1+\sum_{k=1}^{\infty}\frac{1}{k!}
     \int_{(\mathbb{R}^3\times\mathbb{R}^3)^k}d x_3\ldots d x_{k+2}V_{1+k}(\{q_1,q_2\},q_3,\ldots,q_{k+2})
     \prod_{i=3}^{k+2} F_1(t,x_i)\big),\nonumber
\end{eqnarray}
and $V_{1+k}(\{q_1,q_2\},q_3,\ldots,q_{k+2})$ is the sum of all graphs of $k$ labeled points which are
biconnected when the Mayer function $f(q_1,q_2)\equiv\mathcal{X}_2(q_1,q_2)-1$ is added,
$\mathcal{X}_2(q_1,q_2)\equiv\mathcal{X}_{\mathbb{R}^{6}\setminus \mathbb{W}_{2}}$ is the Heaviside step
function of allowed configurations of two hard spheres. Taking into account that the Boltzmann-Enskog collision
integral is defined by the first term of expansion (\ref{ClustExp}), we denote the collision integral
of the revised Enskog equation in the form of the corresponding expansion
\begin{eqnarray}\label{REE}
  &&\mathcal{I}_{REE}\equiv\mathcal{I}_{BEE}+
      \sum_{k=1}^{\infty}(\mathcal{I}^{(k)}_{REE,+}-\mathcal{I}^{(k)}_{REE,-}),
\end{eqnarray}
where
\begin{eqnarray*}
  &&\hskip-7mm\mathcal{I}^{(k)}_{REE,\pm}\doteq\sigma^2\int_{\mathbb{R}^3\times \mathbb{S}^2_+}d p_2 d\eta
    \langle\eta,(p_1-p_2)\rangle \int_{\mathbb{R}^{3k}}d q_3\ldots d q_{k+2}
    V(\{q_1,q_1\pm\sigma\eta\},q_3,\ldots,q_{k+2})\times\\
  &&\hskip-7mm\times F_1(t,q_1,p_1^{\sharp})F_1(t,q_1\pm\sigma\eta,p_2^{\sharp})\prod_{i=2}^{k+2}F_1(t,x_i).
\end{eqnarray*}

We compare term by term collision integral (\ref{REE}) of the revised Enskog equation and the collision
integral (\ref{GEE}) of the Markovian generalized Enskog equation (\ref{MarkGke1}) by the example of the
first terms of these expansions. In case of collision integral (\ref{REE}) this term has the form
\begin{eqnarray}\label{REE1}
  &&\mathcal{I}^{(1)}_{REE,\pm}=\sigma^2\int_{\mathbb{R}^3\times \mathbb{S}^2_+}d p_2 d\eta
     \langle\eta,(p_1-p_2)\rangle\int_{\mathbb{R}^3\times\mathbb{R}^3}dx_3
     f(q_1,q_3)f(q_1\pm\sigma\eta,q_3)\times\\
  && \times F_1(t,q_1,p_1^{\sharp})F_1(t,q_1\pm\sigma\eta,p_2^{\sharp})F_1(t,x_3)=\nonumber\\
  &&=\sigma^2\int_{\mathbb{R}^3\times \mathbb{S}^2_+}d p_2 d\eta
     \langle\eta,(p_1-p_2)\rangle\int_{\mathbb{R}^3\times\mathbb{R}^3}dx_3
     \big(\mathcal{X}_2(q_1,q_3)\mathcal{X}_2(q_1\pm\sigma\eta,q_3)-\nonumber \\
  &&-\mathcal{X}_2(q_1,q_3)-\mathcal{X}_2(q_1\pm\sigma\eta,q_3)+1\big)
     F_1(t,q_1,p_1^{\sharp})F_1(t,q_1\pm\sigma\eta,p_2^{\sharp})F_1(t,x_3),\nonumber
\end{eqnarray}
and in case of collision integral (\ref{GEE}) the corresponding term has the following form
\begin{eqnarray}\label{GEE1}
   &&\hskip-7mm\mathcal{I}^{(1)}_{GEE,\pm}=\sigma^2\int_{\mathbb{R}^3\times\mathbb{S}^2_+}d p_2 d\eta
      \langle\eta,(p_1-p_2)\rangle\int_{\mathbb{R}^3\times\mathbb{R}^3}dx_3\lim\limits_{\varepsilon\rightarrow 0}
      \big(\widehat{\mathfrak{A}}_{2}(\varepsilon^{-1}t,\{1^{\sharp},2^{\sharp}_{\pm}\},3)-\nonumber \\
    &&\hskip-7mm-\widehat{\mathfrak{A}}_{1}(\varepsilon^{-1}t,\{1^{\sharp},2^{\sharp}_{\pm}\})
      \widehat{\mathfrak{A}}_{2}(\varepsilon^{-1}t,1^{\sharp},3)-
      \widehat{\mathfrak{A}}_{1}(\varepsilon^{-1}t,\{1^{\sharp},2^{\sharp}_{\pm}\})
      \widehat{\mathfrak{A}}_{2}(\varepsilon^{-1}t,2^{\sharp}_{\pm},3)\big)\times\\
    &&\hskip-7mm\times F_1(t,q_1,p_1^{\sharp})F_1(t,q_1\pm\sigma\eta,p_2^{\sharp})F_1(t,x_3),\nonumber
\end{eqnarray}
where we use notations accepted in equation (\ref{gke1}) and the evolution operator
$\widehat{\mathfrak{A}}_{n}(\varepsilon^{-1}t)$ is a corresponding order scattering
cumulant (\ref{scacu}).

We transform expression (\ref{GEE1}) to the form adopted to the structure of expression (\ref{REE1}).
With that end in view we introduce the notion of the scattering operator
\begin{eqnarray*}\label{so}
   &&\widehat{S}_n(1,\ldots,n)=\lim\limits_{\varepsilon\rightarrow 0}S_n(-\varepsilon^{-1}t,1,\ldots,n)
      \prod_{i=1}^{n}S_1(\varepsilon^{-1}t,i), \quad n\geq1,
\end{eqnarray*}
where the evolution operator $S_n(-\varepsilon^{-1}t)$ is defined by (\ref{Sspher}).
Then, according to definition (\ref{scacu}) of scattering cumulants and definition (\ref{Sspher}), for the first
summand in expression (\ref{GEE1}) the following equality holds
\begin{eqnarray*}
   &&\lim\limits_{\varepsilon\rightarrow 0}\widehat{\mathfrak{A}}_{2}(\varepsilon^{-1}t,\{1^{\sharp},2^{\sharp}_{\pm}\},3)=
    \big(f(q_1,q_3)f(q_1\pm\sigma\eta,q_3)+f(q_1,q_3)+f(q_1\pm\sigma\eta,q_3)+\\
   &&+1\big)\widehat{S}_3(1^{\sharp},2^{\sharp}_{\pm},3)-
     \lim\limits_{\varepsilon\rightarrow 0}S_2(-\varepsilon^{-1}t,1^{\sharp},2^{\sharp}_{\pm})S_1(-\varepsilon^{-1}t,3)
     \big(f(q_1,q_3)f(q_1\pm\sigma\eta,q_3)+\\
   &&+f(q_1,q_3)+f(q_1\pm\sigma\eta,q_3)+1\big)S_1(\varepsilon^{-1}t,1^{\sharp})
      S_1(\varepsilon^{-1}t,2^{\sharp}_{\pm})S_1(\varepsilon^{-1}t,3),
\end{eqnarray*}
where we take into consideration that $S_{3}(\varepsilon^{-1}t,1^{\sharp},2^{\sharp}_{\pm},3)\big(f(q_1,q_3)f(q_1\pm\sigma\eta,q_3)+f(q_1,q_3)+f(q_1\pm\sigma\eta,q_3)\big)=0$.
For other summands in expression (\ref{GEE1}) it holds, respectively
\begin{eqnarray*}
   &&\lim\limits_{\varepsilon\rightarrow 0}\widehat{\mathfrak{A}}_{1}(\varepsilon^{-1}t,\{1^{\sharp},2^{\sharp}_{\pm}\})
      \widehat{\mathfrak{A}}_{2}(\varepsilon^{-1}t,1^{\sharp},3)=
      \widehat{S}_2(1^{\sharp},2^{\sharp}_{\pm})\big(f(q_1,q_3)+1\big)\widehat{S}_2(1^{\sharp},3)-\\
   &&-\lim\limits_{\varepsilon\rightarrow 0}S_2(-\varepsilon^{-1}t,1^{\sharp},2^{\sharp}_{\pm})S_1(-\varepsilon^{-1}t,3)
     \big(f(q_1,q_3)+1\big)S_1(\varepsilon^{-1}t,1^{\sharp})S_1(\varepsilon^{-1}t,2^{\sharp}_{\pm})S_1(\varepsilon^{-1}t,3),\\ \\
   &&\lim\limits_{\varepsilon\rightarrow 0}\widehat{\mathfrak{A}}_{1}(\varepsilon^{-1}t,\{1^{\sharp},2^{\sharp}_{\pm}\})
      \widehat{\mathfrak{A}}_{2}(\varepsilon^{-1}t,2^{\sharp}_{\pm},3)=
      \widehat{S}_2(1^{\sharp},2^{\sharp}_{\pm})\big(f(q_1\pm\sigma\eta,q_3)+1\big)\widehat{S}_2(2^{\sharp}_{\pm},3)-\\
   &&-\lim\limits_{\varepsilon\rightarrow 0}S_2(-\varepsilon^{-1}t,1^{\sharp},2^{\sharp}_{\pm})S_1(-\varepsilon^{-1}t,3)
     \big(f(q_1\pm\sigma\eta,q_3)+1\big)S_1(\varepsilon^{-1}t,1^{\sharp})S_1(\varepsilon^{-1}t,2^{\sharp}_{\pm})S_1(\varepsilon^{-1}t,3).
\end{eqnarray*}

As a result of the validity of these equalities, taking into account that $\widehat{S}_2(1^{\sharp},2^{\sharp}_{\pm})=I$ is
the identity operator and, owing to the following equality
\begin{eqnarray*}
   &&\int_{\mathbb{R}^3\times\mathbb{R}^3}dx_3\lim\limits_{\varepsilon\rightarrow 0}
     S_2(-\varepsilon^{-1}t,1^{\sharp},2^{\sharp}_{\pm})S_1(-\varepsilon^{-1}t,3)
     f(q_1,q_3)f(q_1\pm\sigma\eta,q_3)\times\\
   &&\times S_1(\varepsilon^{-1}t,1^{\sharp})S_1(\varepsilon^{-1}t,2^{\sharp}_{\pm})S_1(\varepsilon^{-1}t,3)
     F_1(t,q_1,p_1^{\sharp})F_1(t,q_1\pm\sigma\eta,p_2^{\sharp})F_1(t,x_3)=0,
\end{eqnarray*}
expression (\ref{GEE1}) takes the form
\begin{eqnarray}\label{GEE2}
   &&\mathcal{I}^{(1)}_{GEE,\pm}=\sigma^2\int_{\mathbb{R}^3\times\mathbb{S}^2_+}d p_2 d\eta
      \langle\eta,(p_1-p_2)\rangle\int_{\mathbb{R}^3\times\mathbb{R}^3}dx_3
      \big(\mathcal{X}_2(q_1,q_3)\mathcal{X}_2(q_1\pm\sigma\eta,q_3)\times \\
   &&\times \widehat{S}_{3}(1^{\sharp},2^{\sharp}_{\pm},3)-\mathcal{X}_2(q_1,q_3)\widehat{S}_2(1^{\sharp},3)-
       \mathcal{X}_2(q_1\pm\sigma\eta,q_3)\widehat{S}_2(2^{\sharp}_{\pm},3)+1\big)\times \nonumber \\
   &&\times F_1(t,q_1,p_1^{\sharp})F_1(t,q_1\pm\sigma\eta,p_2^{\sharp})F_1(t,x_3).\nonumber
\end{eqnarray}
It is clear that last expression (\ref{GEE2}) coincides with corresponding term (\ref{REE1}) of the collision
integral expansion of the revised Enskog equation, if we neglect the collisions which occur in the process
of the evolution of hard spheres.

As we note above in section 4.1 expansions of marginal functionals of the state (\ref{f}) are
nonequilibrium analog of the Mayer-Ursell expansions over powers of the density of equilibrium
distribution functions \cite{B46},\cite{UF,GP63}. Hence, the structures of expansions of collision
integral (\ref{REE}) of the revised Enskog equation and collision integral (\ref{GEE}) of the
Markovian generalized Enskog equation (\ref{MarkGke1}) are the same, so observations given for
the first terms of these collision integral expansions take place for every term. In fact every
term of collision integral (\ref{GEE}) give us the classification of correlations generated by
pair collisions of hard spheres.

Thus, in case of neglecting of collisions in particle clusters of more than two hard spheres,
the collision integral of the generalized Enskog kinetic equation
(\ref{MarkGke1}) transforms to collision integral (\ref{REE}) of the revised Enskog equation.


\section{The Cauchy problem of the generalized Enskog equation}

\subsection{An existence theorem}
We consider the abstract Cauchy problem (\ref{gke1})-(\ref{2}) ((\ref{gke1n}),(\ref{2})) in the space
of integrable functions $L^{1}(\mathbb{R}^3\times\mathbb{R}^3)$. The following statement is true.
\begin{theorem}
A global in time solution of the Cauchy problem of the generalized Enskog equation (\ref{gke1})-(\ref{2})
((\ref{gke1n}),(\ref{2})) is determined by the expansion
\begin{eqnarray}\label{ske}
  &&\hskip-12mmF_{1}(t,x_1)=\sum\limits_{n=0}^{\infty}\frac{1}{n!}
     \int_{(\mathbb{R}^3\times\mathbb{R}^3)^n}dx_2\ldots dx_{n+1}\,
     \mathfrak{A}_{1+n}(-t,1,\ldots,n+1)\prod_{i=1}^{n+1}F_{1}^0(x_i)
     \mathcal{X}_{\mathbb{R}^{3(1+n)}\setminus \mathbb{W}_{1+n}},
\end{eqnarray}
where the cumulants $\mathfrak{A}_{1+n}(-t),\,n\geq0,$ of groups of operators (\ref{Sspher}) are
defined by formula (\ref{nLkymyl}). If $\|F_1^0\|_{L^{1}(\mathbb{R}^3\times\mathbb{R}^3)}<e^{-10}(1+e^{-9})^{-1}$,
then for $F_1^0\in{L}^{1}_{0}(\mathbb{R}^3\times\mathbb{R}^3)$ it is a strong solution
and for an arbitrary initial data $F_1^{0}\in L^{1}(\mathbb{R}^3\times\mathbb{R}^3)$ it is a weak solution.
\end{theorem}
\begin{proof}
Let $F_1^0\in L^{1}_{0}(\mathbb{R}^3\times\mathbb{R}^3)$. Series (\ref{ske}) converges in the norm
of the space $L^{1}(\mathbb{R}^3\times\mathbb{R}^3)$ provided that (\ref{estCon}). We use the result
of the differentiation over the time variable of expression (\ref{ske}) in the sense of pointwise
convergence on the space $L^{1}(\mathbb{R}^3\times\mathbb{R}^3)$ from section 3.3. Then for
$F_1^0\in L^{1}_{0}(\mathbb{R}^3\times\mathbb{R}^3)$ it holds that
\begin{eqnarray}\label{sts}
  &&\lim_{\Delta t\rightarrow 0}\int_{\mathbb{R}^3\times\mathbb{R}^3}dx_1
     \big|\frac{1}{\Delta t}\big(F_{1}(t+\Delta t,x_1)-F_{1}(t,x_1)\big)-\\
  &&-\big(-\langle p_1,\frac{\partial}{\partial q_1}\rangle F_{1}(t,x_1)
     +\sum\limits_{n=0}^{\infty}\frac{1}{n!}\int_{\mathbb{R}^3\times\mathbb{R}^3}dx_2
     \mathcal{L}^*_{\mathrm{int}}(1,2)\times \nonumber \\
  &&\times\int_{(\mathbb{R}^{3}\times\mathbb{R}^{3})^{n}}dx_3\ldots dx_{n+2}
     \mathfrak{V}_{1+n}\big(t,\{1,2\},3,\ldots,n+2\big)\prod_{i=1}^{n+2}F_{1}(t,x_i)\big)\big|=0,\nonumber
\end{eqnarray}
where the operator $\mathcal{L}^*_{\mathrm{int}}(1,2)$ is defined by formula (\ref{aLint}) and abridged
notations are used: the symbols $F_{1}(t,x_1)$ and $\prod _{i=1}^{n+2}F_{1}(t,x_i)$ are
implied series (\ref{ske}) and for $s=2$ series (\ref{Pr}), respectively.

The validity of equality (\ref{sts}) means that a strong solution of initial-value problem (\ref{gke1})-(\ref{2})
is given by series (\ref{ske}).

Let us prove that in case of initial data $F_1^{0}\in L^{1}(\mathbb{R}^3\times\mathbb{R}^3)$
a weak solution of the initial-value problem of the generalized Enskog equation (\ref{gke1})
is represented by series (\ref{ske}). With this purpose we introduce the functional
\begin{eqnarray}\label{funcc}
  &&(\varphi_{1},F_{1}(t))\doteq\int_{\mathbb{R}^3\times\mathbb{R}^3}dx_1\varphi_{1}(x_1)F_{1}(t,x_1),
\end{eqnarray}
where $\varphi_{1}\in L_0(\mathbb{R}^3\times\mathbb{R}^3)$
is continuously differentiable bounded function with compact support and the function $F_{1}(t)$ is
defined by series (\ref{ske}) which is convergent series in the norm of the space
$L^{1}(\mathbb{R}^3\times\mathbb{R}^3)$  under condition (\ref{estCon}). This functional exists since
the function $F_1(t)$ is integrable and $\varphi_1$ is a bounded function.

Using expansion (\ref{ske}), we transform  functional (\ref{funcc}) as follows
\begin{eqnarray*}\label{funk-gN}
  &&(\varphi_{1},F_{1}(t))=\\
  &&=\sum\limits_{n=0}^{\infty}\frac{1}{n!}\,
      \int_{(\mathbb{R}^3\times\mathbb{R}^3)^{n+1}}\,dx_1\ldots dx_{n+1}\varphi_{1}(x_1)
      \mathfrak{A}_{1+n}(-t,1,\ldots,n+1)\prod_{i=1}^{n+1}F_1^0(x_i)
      \mathcal{X}_{\mathbb{R}^{3(1+n)}\setminus \mathbb{W}_{1+n}}=\\
   &&=\sum\limits_{n=0}^{\infty}\frac{1}{n!}\int_{(\mathbb{R}^3\times\mathbb{R}^3)^{n+1}}
      dx_1\ldots dx_{n+1}\mathfrak{A}_{1+n}(t,1,\ldots,n+1)\varphi_{1}(x_1)
      \prod_{i=1}^{n+1}F_{1}^0(x_i)\mathcal{X}_{\mathbb{R}^{3(1+n)}\setminus\mathbb{W}_{1+n}},
\end{eqnarray*}
where the cumulant $\mathfrak{A}_{1+n}(t),\,n\geq0$ is adjoint to the cumulant $\mathfrak{A}_{1+n}(-t),\,n\geq0,$
in the sense of functional (\ref{funcc}). For $F_1^{0}\in L^{1}(\mathbb{R}^3\times\mathbb{R}^3)$ and
$\varphi_{1}\in L_0(\mathbb{R}^3\times\mathbb{R}^3)$ in the sense of the $\ast$-weak convergence on the space
$L(\mathbb{R}^3\times\mathbb{R}^3)$ the following equality holds
\begin{eqnarray}\label{d_funk-d}
  &&\lim\limits_{\Delta t\rightarrow0}\sum\limits_{n=0}^{\infty}\frac{1}{n!}\,
     \int_{(\mathbb{R}^3\times\mathbb{R}^3)^{n+1}}\,dx_1\ldots dx_{n+1}
     \frac{1}{\Delta t}\big(\mathfrak{A}_{1+n}(t+\Delta t)-\\
  &&-\mathfrak{A}_{1+n}(t)\big)\varphi_{1}(x_1)\prod_{i=1}^{n+1}F_1^0(x_i)
     \mathcal{X}_{\mathbb{R}^{3(1+n)}\setminus \mathbb{W}_{1+n}}=\nonumber
\end{eqnarray}
\begin{eqnarray*}  
  &&=(\mathcal{L}_1\varphi_{1},F_{1}(t))+
     \sum\limits_{n=1}^{\infty}\frac{1}{n!}\int_{(\mathbb{R}^3\times\mathbb{R}^3)^{n+2}}dx_1\ldots dx_{n+2}
     \mathcal{L}_{\mathrm{int}}(1,2)\varphi_{1}(x_1)\times\nonumber\\
  &&\times\mathfrak{A}_{1+k}(-t,\{1,2\},3,\ldots,n+2)\prod_{i=1}^{n+2}F_{1}^0(x_i)
     \mathcal{X}_{\mathbb{R}^{3(n+2)}\setminus \mathbb{W}_{n+2}}.\nonumber
\end{eqnarray*}
In obtained expression (\ref{d_funk-d}) the operator $\mathcal{L}_{\mathrm{int}}(1,2)$ is adjoint to the operator
$\mathcal{L}^{*}_{\mathrm{int}}(1,2)$ and it is defined on $L_0(\mathbb{R}^{6}\times(\mathbb{R}^{6}\setminus \mathbb{W}_2))$
by the formula
\begin{eqnarray}\label{dLint}
  &&\hskip-5mm \int_{(\mathbb{R}^3\times\mathbb{R}^3)^2}dx_1dx_2
     \mathcal{L}_{\mathrm{int}}(1,2)\varphi_{2}(x_1,x_2)\doteq
     \sigma^{2}\int_{\mathbb{R}^3\times\mathbb{R}^3}dx_1\int_{\mathbb{R}^3\times\mathbb{S}_{+}^{2}}d p_2
     d\eta\langle\eta,(p_1-p_{2})\rangle \times\\ \nonumber\\
  &&\hskip-5mm \times\big(\varphi_{2}(q_1,p_1^*,q_1+\sigma\eta,p_{2}^*)-\varphi_{2}(q_1,p_1,q_1+\sigma\eta,p_{2})\big).\nonumber
\end{eqnarray}
For $F_1^{0}\in L^{1}(\mathbb{R}^3\times \mathbb{R}^3)$ and
$\varphi_{1}\in L_0(\mathbb{R}^3\times\mathbb{R}^3)$  the limit functionals (\ref{d_funk-d}) exist.
Using equality (\ref{Pr}), we transform the second functional in (\ref{d_funk-d}) to the form
\begin{eqnarray}\label{d_funk-dn}
  &&\sum\limits_{n=1}^{\infty}\frac{1}{n!}
     \int_{(\mathbb{R}^3\times\mathbb{R}^3)^{n+2}}dx_1\ldots dx_{n+2}\mathcal{L}_{\mathrm{int}}(1,2)
     \varphi_{1}(x_1)\mathfrak{A}_{1+k}(-t,\{1,2\},3,\ldots,n+2)\times  \\
  &&\times\prod _{i=1}^{n+2}F_{1}^0(x_i)\mathcal{X}_{\mathbb{R}^{3(n+2)}\setminus \mathbb{W}_{n+2}}
     =\sum\limits_{n=0}^{\infty}\frac{1}{n!}\int_{(\mathbb{R}^3\times\mathbb{R}^3)^{2}}dx_1 dx_2
     \mathcal{L}_{\mathrm{int}}(1,2)\varphi_1(x_1)\times\nonumber \\
  &&\times\int_{(\mathbb{R}^3\times\mathbb{R}^3)^{n}}dx_3\ldots dx_{n+2}
     \mathfrak{V}_{1+n}\big(t,\{1,2\},3,\ldots,n+2\big)\prod_{i=1}^{n+2} F_{1}(t,x_i).\nonumber
\end{eqnarray}

Therefore, in consequence of equalities (\ref{d_funk-d}),(\ref{d_funk-dn}), for functional
(\ref{funcc}) we obtain the equality
\begin{eqnarray*}
  &&\hskip-7mm\frac{d}{dt}(\varphi_{1},F_{1}(t))=(\mathcal{L}_1\varphi_{1},F_{1}(t))+\\
  &&\hskip-7mm+\int_{(\mathbb{R}^{3}\times\mathbb{R}^{3})^2}dx_1 dx_2\mathcal{L}_{\mathrm{int}}(1,2)\varphi_{1}(x_1)
    \sum\limits_{n=0}^{\infty}\frac{1}{n!}\int_{(\mathbb{R}^3\times\mathbb{R}^3)^{n}}
    dx_3\ldots dx_{n+2}\mathfrak{V}_{1+n}(t)\prod_{i=1}^{n+2}F_{1}(t,i),\nonumber
\end{eqnarray*}
or in the explicit form
\begin{eqnarray}\label{d_funk-gNr}
  &&\hskip-8mm\frac{d}{dt}(\varphi_{1},F_{1}(t))=\int_{\mathbb{R}^{3}\times\mathbb{R}^{3}}dx_1
     \langle p_1,\frac{\partial}{\partial q_1}\rangle\varphi_{1}(x_1)F_{1}(t,x_1)+\\
  &&\hskip-8mm+\sigma^{2}\int_{\mathbb{R}^{3}\times\mathbb{R}^{3}}dx_1
     \int_{\mathbb{R}^{3}\times\mathbb{S}_+^2}d p_2d\eta \langle\eta,p_1-p_2\rangle(\varphi_1(q_1,p_1^*)-\varphi_1(q_1,p_1))
     \sum\limits_{n=0}^{\infty}\frac{1}{n!}\int_{(\mathbb{R}^{3}\times\mathbb{R}^{3})^{n}}dx_3\ldots\nonumber\\
  &&\hskip-8mm\ldots dx_{n+2}\mathfrak{V}_{1+n}(t,\{1,2_{+}\},3,\ldots,n+2)F_{1}(t,x_1)F_{1}(t,q_1+\sigma\eta,p_2)
      \prod_{i=3}^{n+2}F_{1}(t,x_i).\nonumber
\end{eqnarray}
This equality (\ref{d_funk-gNr}) means that series (\ref{ske}) for arbitrary initial data
$F_1^{0}\in L^{1}(\mathbb{R}^3\times \mathbb{R}^3)$ is a weak solution of the Cauchy problem
of the generalized Enskog equation (\ref{gke1})-(\ref{2}).

\end{proof}

\subsection{A weak solution in the extended sense}
For the Cauchy problem (\ref{gke1})-(\ref{2}) we introduced the notion of a weak solution
in the extended sense. Let $\big(\varphi,F(t\mid F_{1}(t))\big)$ be the functional
\begin{eqnarray}\label{func-g}
  &&\hskip-8mm\big(\varphi,F(t\mid F_{1}(t))\big)\doteq \sum_{s=0}^{\infty}\frac{1}{s!}
      \int_{(\mathbb{R}^3\times\mathbb{R}^3)^{s}}dx_1\ldots dx_{s}
      \varphi_{s}(x_1,\ldots,x_s)F_{s}(t,x_1,\ldots,x_s\mid F_{1}(t)),
\end{eqnarray}
where $\varphi=(\varphi_0,\varphi_1(x_1),\ldots,\varphi_s(x_1,\ldots,x_s),\ldots)$ is a finite
sequence of bounded infinitely times differentiable functions with compact supports
$\varphi_s\in L_{0}(\mathbb{R}^{3s}\times(\mathbb{R}^{3s}\setminus \mathbb{W}_s))\in
L(\mathbb{R}^{3s}\times(\mathbb{R}^{3s}\setminus \mathbb{W}_s)),\,s\geq0$,
and elements of the sequence $F(t,\mid F_{1}(t))\doteq\big(1,F_{1}(t,x_1),F_{2}(t,x_1,x_2\mid F_{1}(t)),
\ldots,$ $F_{s}(t,x_1,\ldots,x_s \mid F_{1}(t)),\ldots\big)$ are defined by formulas (\ref{ske})
and (\ref{f}) for the first and other elements respectively. Owing to estimates (\ref{estFunc12})
and (\ref{inequalF}), this functional exists provided that:
$\|F_1^0\|_{L^{1}(\mathbb{R}^{3}\times\mathbb{R}^{3})}<e^{-(3s+4)}(1+e^{-(3s+3)})^{-1}$.

On functions $\varphi_s\in L_{0}(\mathbb{R}^{3s}\times(\mathbb{R}^{3s}\setminus \mathbb{W}_s)),\,s\geq1$,
it is defined the operator ${\mathcal{B}}^{+}$ which is adjoint to the generator of the BBGKY hierarchy
of hard spheres (\ref{NelBog1}) for $t>0$,
\begin{eqnarray*}\label{wg}
  &&({\mathcal{B}}^{+}\varphi)_{s}(x_1,\ldots,x_s)\doteq \mathcal{L}_{s}(Y)\varphi_{s}(x_1,\ldots,x_s)+
       \sum_{j_1\neq j_{2}=1}^s\mathcal{L}_{\mathrm{int}}(j_1,j_{2})
       \varphi_{s-1}((x_1,\ldots,x_s)\backslash(x_{j_1})),
\end{eqnarray*}
where the operator $\mathcal{L}_{\mathrm{int}}(j_1,j_{2})$ is defined by (\ref{dLint}) and the Liouville
operator $\mathcal{L}_s(Y)$ is adjoint to operator (\ref{OperL}) in the sense of functional (\ref{func-g})
\cite{CGP97}.

It is said to be that expansion (\ref{ske}) is a weak solution of the Cauchy problem (\ref{gke1})-(\ref{2})
in the extended sense, if for functional (\ref{func-g}) the following equality is true
\begin{eqnarray}\label{w}
  &&\frac{d}{dt}\big(\varphi,F(t\mid F_{1}(t))\big)=\big({\mathcal{B}}^{+}\varphi,F(t\mid F_{1}(t))\big).
\end{eqnarray}
In view of estimates (\ref{inequalF}) the functional on the right-hand side of this equation exists
provided that (\ref{estFunc1}).

To verify this definition we transform functional (\ref{func-g}) as follows \cite{GerUJP}
\begin{eqnarray*}
  &&\hskip-9mm\big(\varphi,F(t\mid F_{1}(t))\big)= \\
  &&\hskip-9mm=\sum_{s=0}^{\infty}\,\frac{1}{s!}\int_{(\mathbb{R}^3\times\mathbb{R}^3)^{s}}dx_1\ldots dx_{s}
     \sum_{n=0}^s\,\frac{1}{(s-n)!}\sum_{j_1\neq\ldots\neq j_{s-n}=1}^s\,
     \sum\limits_{Z\subset Y\backslash (j_1,\ldots,j_{s-n})}(-1)^{|Y\backslash (j_1,\ldots,j_{s-n})\backslash Z|}\times\\
  &&\hskip-9mm\times S_{s-n+|Z|}(t,(j_1,\ldots,j_{s-n})\cup Z )\varphi_{s-n}(x_{j_1},\ldots,x_{j_{s-n}})
     \prod_{i=1}^{s}F_{1}^0(x_i)\mathcal{X}_{\mathbb{R}^{3s}\setminus\mathbb{W}_{s}},
\end{eqnarray*}
where ${\sum\limits}_{Z\subset Y\backslash (j_1,\ldots,j_{s-n})}$ is a sum over all subsets
$Z\subset Y\backslash (j_1,\ldots,j_{s-n})$ of the set $Y\backslash (j_1,\ldots,j_{s-n})\subset Y\equiv(1,\ldots,s)$.
For $F_1^{0}\in L^{1}(\mathbb{R}^{3}\times\mathbb{R}^{3})$ and
$\varphi_s\in L_{0}(\mathbb{R}^{3s}\times(\mathbb{R}^{3s}\setminus \mathbb{W}_s)),\,s\geq1$, this functional exists.

Skipping the details, as a result for $\varphi_s\in L_{0}(\mathbb{R}^{3s}\times(\mathbb{R}^{3s}\setminus \mathbb{W}_s)),\,s\geq1$,
the derivative of functional (\ref{func-g}) over the time variable in the sense of the $\ast$-weak convergence
on the space of sequences of bounded functions takes the form
\begin{eqnarray}\label{d_funk-gN}
  &&\frac{d}{dt}\big(\varphi,F(t\mid F_{1}(t))\big)=
     \sum\limits_{s=0}^{\infty}\frac{1}{s!}\int_{(\mathbb{R}^3\times\mathbb{R}^3)^{s}}dx_1\ldots dx_{s}
     \big(\mathcal{L}_{s}(Y)\varphi_{s}(x_1,\ldots,x_{s})+\\
  &&+\sum_{j_1\neq j_{2}=1}^s\mathcal{L}_{\mathrm{int}}(j_1,j_{2})
     \varphi_{s-1}((x_1,\ldots,x_{s})\backslash (x_{j_1}))\big)
     F_{s}(t,x_1,\ldots,x_{s}\mid F_{1}(t)).\nonumber
\end{eqnarray}
According to the notion of a weak solution in the extended sense (\ref{w}), equality
(\ref{d_funk-gN}) means that for arbitrary initial data $F_1^{0}\in L^{1}(\mathbb{R}^{3}\times\mathbb{R}^{3})$
a weak solution of initial-value problem of the generalized Enskog equation (\ref{gke1}) is determined
by series (\ref{ske}). In fact, we have also proved that marginal functionals of the state (\ref{f})
are weak solutions of the BBGKY hierarchy of hard spheres.


\section{Conclusion}
Thus, in case of initial data (\ref{eq:Bog2_haos}) solution (\ref{F(t)}) of the Cauchy problem
(\ref{NelBog1})-(\ref{eq:NelBog2}) of the BBGKY hierarchy for hard spheres and a solution of the
Cauchy problem of the generalized Enskog equation (\ref{gke1})-(\ref{2}) together with marginal
functionals of the state (\ref{f}) give two equivalent approaches to the description of the evolution
of a hard sphere system. We emphasize that in addition, the coefficients of the generalized Enskog
equation (\ref{gke1}) and marginal functionals of the state (\ref{f}) are determined by the operators
of initial correlations specified by the forbidden configurations of hard spheres.

The structure of the collision integral expansion of the generalized Enskog equation (\ref{gke1}) is such
that the first term of this expansion is the Boltzman-Enskog collision integral and the next terms describe
all possible correlations which are created by hard sphere dynamics and by the propagation of initial
correlations connected with the forbidden configurations.

On the basis of the derived generalized Enskog equation (\ref{gke1}) we state the Markovian generalized
Enskog equation (\ref{MarkGke1}). In case of absence of correlations which generated by hard sphere
dynamics the collision integral of this kinetic equation transforms to the collision integral of the
revised Enskog equation \cite{BLPT}. Therefore the kinetic evolution of hard spheres described by the
revised Enskog equation makes it possible to take into consideration only the initial correlations
specified by the forbidden configurations of particles.

In the end it should be emphasized that the kinetic evolution is an inherent property of infinite-particle
systems. In spite of the fact that in terms of a one-particle marginal distribution function from the space
of integrable functions can be described a hard sphere system with the finite average number of particles,
the Enskog kinetic equation has been derived on the basis of the formalism of nonequilibrium grand canonical
ensemble since its framework is adopted to the description of infinite-particle systems in suitable functional
spaces \cite{CGP97} as well.

We remark also that developed approach is also related to the problem of a rigorous derivation of the
non-Markovian kinetic-type equations from underlaying many-particle dynamics which make possible to describe
the memory effects of particle and energy transport, for example, the anomalous transport in turbulent plasma,
the Brownian motion of macroparticles in complex fluids.


\addcontentsline{toc}{section}{References}
\renewcommand{\refname}{References}

\end{document}